\documentclass[a4paper,10pt,reqno]{amsart}

\usepackage{hyperref}
\usepackage{amssymb}
\usepackage{eucal}
\usepackage{eufrak}
\usepackage{graphicx}
\usepackage{epsfig}
\usepackage{multicol}
\usepackage[all]{xy}
\usepackage{amsmath}
\usepackage{wasysym}
\usepackage{dsfont}
\usepackage{color}

\usepackage{framed}

\usepackage{amsthm}
\theoremstyle{plain} 
\newtheorem{theorem}{Theorem}[section]

\newtheorem{lemma}[theorem]{Lemma}

\theoremstyle{definition}  
\newtheorem{definition}[theorem]{Definition}

\newtheorem{remark}[theorem]{Remark}

\let\originalleft\left
\let\originalright\right
\renewcommand{\left}{\mathopen{}\mathclose\bgroup\originalleft}
\renewcommand{\right}{\aftergroup\egroup\originalright}

\newcommand{\lp}{\left(}
\newcommand{\rp}{\right)}
\newcommand{\lc}{\left\{}
\newcommand{\rc}{\right\}}
\newcommand{\der}{\partial}
\newcommand{\cua}{^{2}}

\newcommand{\R}{\mathds{R}}      
\newcommand{\F}{\mathds{F}}
\newcommand{\I}{\mathds{I}}
\newcommand{\Flder}{\rightarrow}
\newcommand{\lcf}{\lbrack\! \lbrack}
\newcommand{\rcf}{\rbrack\! \rbrack}

\usepackage{amssymb}

\newcommand{\proa}{A^*G \mbox{$\;$}_{\tau^*} \kern-3pt\times_\alpha
G \mbox{$\;$}_\beta \kern-3pt\times_{\tau^*} A^*G}

\newcommand{\lvec}[1]{\overleftarrow{#1}}
\newcommand{\rvec}[1]{\overrightarrow{#1}}
\newcommand{\alg}{\mathfrak{so}(3)}
\newcommand{\dalg}{\mathfrak{so}^{*}(3)}
\newcommand{\gru}{SO(3)}

\newcommand{\e}{\operatorname{exp}}
\newcommand{\Ad}{\operatorname{Ad}}
\newcommand{\ca}{\operatorname{cay}}
\newcommand{\ad}{\operatorname{ad}}
\newcommand{\spanop}{\operatorname{span}}  

\newcommand{\al}{\mathfrak{g}}
\newcommand{\dal}{\mathfrak{g}^{*}}

\newcommand{\met}{\mathcal{G}}
\newcommand{\dist}{\mathcal{D}}
\newcommand{\qu}{\mathcal{Q}}
\newcommand{\pe}{\mathcal{P}}

\begin{document}

\title{New developments on the Geometric Nonholonomic Integrator}

\author[S. Ferraro]{Sebasti\'an Ferraro}
\address{S. Ferraro: Universidad Nacional del Sur, Instituto de Matem\'atica Bah\'{\i}a Blanca and CONICET, Argentina} \email{sferraro@uns.edu.ar}

\author[F. Jim\'enez]{Fernando Jim\'enez}
\address{F. Jim\'enez: Zentrum Mathematik der Technische Universit\"at M\"unchen, D-85747 Garching bei M\"unchen, Germany} \email{fjimenez@ma.tum.de}

\author[D.\ Mart\'{\i}n de Diego]{David Mart\'{\i}n de Diego}
\address{D.\ Mart\'{\i}n de Diego: Instituto de Ciencias Matem\'aticas, CSIC-UAM-UC3M-UCM,
Campus de Cantoblanco, UAM,
C/Nicol\'as Cabrera, 15
 28049 Madrid, Spain} \email{david.martin@icmat.es}

\thanks{This work has been partially supported by MEC (Spain)
Grants   MTM 2010-21186-C02-01, MTM2009-08166-E,  the ICMAT Severo Ochoa project SEV- 2011-0087 and IRSES-project ``Geomech-246981''. S. Ferraro also wishes to thank CONICET Argentina (PIP 2010--2012 GI 11220090101018), ANPCyT Argentina (PICT 2010-2746) and SGCyT UNS for financial support. The research of F. Jim\'enez has been
supported in its final stage by the DFG Collaborative Research Center TRR 109, ``Discretization in Geometry
and Dynamics''.  The authors would like to thank the reviewers for their helpful remarks.}



\begin{abstract}
In this paper, we will discuss new developments regarding the Geometric Nonholonomic Integrator (GNI) \cite{SDD1,SDD2}.
GNI is a  discretization scheme adapted to
  nonholonomic mechanical systems through a discrete geometric
  approach. This method was designed to account for some of the special geometric structures associated to a  nonholonomic motion, like preservation of energy, preservation of constraints or the nonholonomic momentum equation.
  First, we study the GNI versions of the symplectic-Euler methods, paying special attention to their convergence behavior. Then, we construct an extension of the GNI in the case of affine constraints. Finally, we generalize the proposed method to nonholonomic reduced systems, an important subclass of examples in nonholonomic dynamics. We illustrate the behavior of the proposed method with the example of the Chaplygin sphere, which accounts for the last two features, namely it is both a reduced and an affine system.
\end{abstract}

\maketitle

\section{Introduction}

Nonholonomic constraints have been a subject of deep analysis  since
the dawn of Analytical Mechanics.
The origin of its study is nicely explained in the introduction of the
book by Neimark and Fufaev \cite{Ball},

\begin{quote}
``The birth of the theory of dynamics of nonholonomic systems occurred
at the time when the universal and brilliant analytical formalism
created by Euler and Lagrange was found, to general amazement,
to be inapplicable to the very simple mechanical problems of
rigid bodies rolling without slipping on a plane. Lindel{\"{o}}f's error, detected
by Chaplygin, became famous and rolling systems attracted
the attention of many eminent scientists of the time...''
\end{quote}
Many authors have recently shown a new interest in that theory and
also in its relationship to the new developments in control theory
and robotics. The main characteristic of this last period  is that nonholonomic systems are studied from a geometric perspective (see L.D. Fadeev and A.M. Vershik
\cite{VF} as an advanced and fundamental reference, and also,
\cite{Bl2003,BlKrMaMu1996,Bolsinov,CCMD,Cort,Koiller,Kozlov2,LD} and references
therein). From this perspective, nonholonomic mechanics forms part of a wider body of research called {\em Geometric Mechanics}.

A nonholonomic system is a mechanical system subjected to constraint  functions which are,
roughly speaking, functions on the velocities that are not derivable
from position constraints. They arise, for instance, in mechanical systems that have rolling or certain kinds of sliding contact.
Traditionally, the equations of motion for nonholonomic mechanics
are derived from the  Lagrange-d'Alembert principle, which restricts
the set of infinitesimal variations (or constrained forces) in terms
of the constraint functions.
In such  systems, some differences between unconstrained classical Hamiltonian and Lagrangian sytems and nonholonomic dynamics  appear. For instance,  nonholonomic systems are non-variational in the classical sense, since they arise from the Lagrange-d'Alembert principle and not from Hamilton's principle. Moreover, when the nonholonomic constrains are linear in velocities and a symmetry arises,  energy is preserved but in general momentum is not. Nonholonomic systems are described by an almost-Poisson structure (i.e., there is a bracket that together with the energy on the phase space defines the motion, but the bracket generally does not satisfy the Jacobi identiy); and finally, unlike the Hamiltonian setting, volume may not be preserved in the phase space, leading to interesting asymptotic stability in some cases, despite energy conservation. 

From the applied point of view, in the last decade great attention has been put upon the study of the dynamical behavior  of some particular examples of nonholonomic systems; more concretely, rolling without slipping and spinning of different rigid bodies on a plane or on a sphere. Besides, a hierarchy has been constructed in terms of the  body's surface geometry and mass distribution. The existence of an invariant measure and Hamiltonization of such systems, and the necessary conditions for this existence have been carefully studied in \cite{BM1,BMB,BMK,KoEh}.
See \cite{Bl2003,BlKrMaMu1996,CCMD,Cort,Koiller,LD,VF} for more details about nonholonomic systems.

Recent works, firstly iniciated by J. Cort\'es and S. Mart{\'\i}nez in their seminal paper \cite{Cortes_Martinez:Non-holonomic_integrators}, where the authors introduce the notion of discrete Lagrange-d'Alembert's principle, have been devoted to derive numerical methods for nonholonomic systems (see \cite{Fedorov_Zenkov:Discrete_nonholonomic_LL_systems_on_Lie_groups,IMMM,McLPerl,KMS}). These  numerical integrators for nonholonomic
systems have  very good energy behavior in simulations  and additional properties such as the
preservation of the discrete nonholonomic momentum map. In a different direction,
some of the authors of this paper have introduced
the Geometric Nonholonomic Integrator (GNI), whose properties and original motivations can be found in  \cite{SDD1}, while some of its applications and numerical performance can be found in \cite{SDD2,KFMD}. Particularly, in \cite{KFMD} we have 
examined numerically the geometric nonholonomic integrator (GNI) and the reduced d'Alembert-Pontryagin integrator (RDP) in some typical examples of nonholonomic mechanics: the Chaplygin sleigh and the snakeboard. In a different approach, numerical schemes based on the Hamiltonization of nonholonomic systems have been explored in \cite{BFO,MeBlFe}. Although these methods have shown an excellent qualitative and quantitive behavior, they are quite difficult to implement  with generality since they involve solving a difficult task: the Hamiltonization or an inverse problem for a nonholonomic system \cite{BFM}.

Our aim in this work is to analyze further developments of the GNI method introduced in the mentioned references. Particularly, we focus on two aspects: the GNI extension of the usual symplectic-Euler methods (we prove their consistency order and the fact that they are the adjoint of one another, and the generalization of the method to new situations, namely the cases of affine constraints (definition \ref{GNIAff}), reduction by a Lie group of symmetries (definition \ref{GNIRed}) and Lie algebroids (definition \ref{GNIAlg}). All the new generalizations are appropriately illustrated with theoretical and numerical results. 

The paper is structured as follows: $\S$\ref{Riemm} is devoted to introduce the continuous nonholonomic problem with linear constraints, to obtain the nonholonomic equations by means of the Lagrange-d'Alembert principle and to show how these equations can be reobtained through a projection procedure when the system is endowed with a Riemannian metric. $\S$\ref{section2} summarizes the general theory of variational integrators, while $\S$\ref{section4} presents the proposed GNI integrator. In \S\ref{SymEuExt}, the GNI versions of the symplectic-Euler methods are obtained and their convergence behavior studied in theorem \ref{Teo1}. It is also proved in theorem \ref{adThe} that both methods are adjoint of each other; this fact establishes an interesting parallelism with the free (meaning {\it unconstrained}) variational integrators. $\S$\ref{AffS} accounts for the affine extension of the GNI integrator which is illustrated with the theoretical result of SHAKE and RATTLE methods. Section $\S$\ref{RS} is devoted to the development of the GNI integrator for reduced systems, in the case of both linear and affine constraints. The former case is illustrated with the theoretical result of RATTLE algorithm while the latter (which is also affine) is carefully treated in the example of the Chaplygin sphere with three different moments of inertia, including some numerical results. Finally, in $\S$\ref{LieAlg} we extend the GNI integrator to Lie algebroids.

\section{Continuous nonholonomic mechanics}
\label{Riemm}

Mathematically, the nonholonomic setting can be described as follows. We shall start with a configuration space $Q$, which is an $n$-dimensional differentiable manifold with local coordinates denoted by $q^{i},\hspace{1mm} i=1,...,n=\dim Q$, and a non-integrable distribution $\mathcal{D}$ on $Q$ that describes the linear nonholonomic constraints. We can consider this constant-rank distribution $\mathcal{D}$ as a vector subbundle of the tangent bundle $TQ$ (velocity phase space) of the configuration space. Moreover, and as we mentioned in the introduction, $\mathcal{D}$ defines a set of constraints on the velocities. Locally, the linear constraints are written as follows:
\begin{equation}\label{LC}
\phi^{a}\lp q,\dot q\rp=\mu^{a}_{i}\lp q\rp\dot q^{i}=0,\hspace{2mm} 1\leq a\leq m,
\end{equation}
where $\mbox{rank}\lp\mathcal{D}\rp=n-m$. The annihilator $\mathcal{D}^{\circ}$ is locally given by
\[
\mathcal{D}^{\circ}=\mbox{span}\lc\mu^{a}=\mu_{i}^{a}(q)\,dq^{i};\hspace{1mm} 1\leq a\leq m\rc,
\]
where the 1-forms $\mu^{a}$ are independent.

In addition to the distribution, we need to specify the dynamical evolution of the system, usually by fixing a Lagrangian function $L:TQ\Flder\R$. In nonholonomic mechanics, the procedure permitting the extension from the Newtonian point of view to the Lagrangian one is given by the Lagrange-d'Alembert principle. This principle states that a curve $q:I\subset \R\Flder Q$  is an admissible motion of the system if
\[
\delta\mathcal{J}=\delta\int^{T}_{0}L\lp q\lp t\rp, \dot q\lp t\rp\rp dt=0
\]
for all variations such that $\delta q\lp t\rp\in\mathcal{D}_{q\lp t\rp}$, $0\leq t\leq T$, and if the velocity of the curve itself satisfies the constraints. It is remarkable that the Lagrange-d'Alembert principle is not variational since we are imposing the constraints on the curve ``after extremizing'' the functional $\mathcal{J}$. From Lagrange-d'Alembert's principle, we arrive to the nonholonomic equations
\begin{subequations}\label{LdAeqs}
\begin{align}
\frac{d}{dt}\lp\frac{\der L}{\der\dot q^{i}}\rp-\frac{\der L}{\der q^{i}}&=\lambda_{a}\mu^{a}_{i},\label{Con-1}\\\label{Con-2}
\mu_{i}^{a}(q)\,\dot q^{i}&=0,
\end{align}
\end{subequations}
where $\lambda_{a},\hspace{1mm} a=1,...,m$ is a set of Lagrange multipliers. The right-hand side of equation (\ref{Con-1}) represents the force induced by the constraints, and equations (\ref{Con-2}) represent the constraints themselves.

Now we are going to restrict ourselves to the case of nonholonomic mechanical systems with mechanical Lagrangian, i.e,
\begin{equation}\label{Mec}
L\lp v_{q}\rp=\frac{1}{2}{\mathcal G}\lp v_{q},v_{q}\rp-V\lp q\rp,\hspace{3mm} v_{q}\in T_{q}Q;
\end{equation}
where ${\mathcal G}$ is a Riemannian metric on the configuration space $Q$ locally determined by the matrix $M=\lp {\mathcal G}_{ij}\rp_{1\leq i,j\leq n}$, where ${\mathcal G}_{ij}={\mathcal G}\lp\der/\der q^{i},\der/\der q^{j}\rp$. Using some basic tools of Riemannian geometry (see, for instance, \cite{bullolewis}), we may write the equations of motion of the unconstrained system determined by $L$ as
\begin{equation}\label{EqsRiemanLibre}
\nabla_{\dot c\lp t\rp}\dot c\lp t\rp=-\mbox{grad}\hspace{2mm}V\lp c\lp t\rp\rp,
\end{equation}
where $\nabla$ is the Levi-Civita connection associated with ${\mathcal G}$. Observe that if $V\equiv0$ then the Euler-Lagrange equations become the geodesic equations for the Levi-Civita connection.

When the system is subjected to nonholonomic constraints, the equations turn out to be
\begin{equation}\nonumber
\nabla_{\dot c\lp t\rp}\dot c\lp t\rp=-\mbox{grad}\hspace{2mm}V\lp c\lp t\rp\rp+\lambda\lp c\lp t\rp\rp, \hspace{2mm} \dot c\lp t\rp\in\mathcal{D}_{c\lp t\rp},
\end{equation}
where $\lambda$ is a section of $\mathcal{D}^{\perp}$ along $c$ (see \cite{Bl2003,bullolewis,CCMD}). Here, $\mathcal{D}^{\perp}$ stands for the orthogonal complement of $\mathcal{D}$ with respect to ${\mathcal G}$.

Since $Q$ is equipped with a Riemannian metric,  we can decompose the tangent bundle as $TQ=\mathcal{D}\oplus\mathcal{D}^{\perp}$. Moreover, we can also construct two complementary projectors
$\mathcal{P}:TQ\Flder\mathcal{D},
\ \mathcal{Q}:TQ\Flder\mathcal{D}^{\perp}.
$.
In order to obtain a local expression for $\mathcal{P}$ and $\mathcal{Q}$,  define the vector fields $Z^{a}\hspace{2mm},\,1\leq a\leq m$, on $Q$ by
\[
{\mathcal G}\lp Z^{a},Y\rp=\mu^{a}\lp Y\rp,\hspace{2mm} \mbox{for all}\hspace{2mm} Y\in {\mathfrak X}(M);
\]
that is, $Z^a$ is the gradient vector field of the 1-form $\mu^a$. Thus, $\mathcal{D}^{\perp}$ is spanned by $Z^a$, $1\leq a\leq m$. In local coordinates:
\[
Z^{a}={\mathcal G}^{ij}\mu^{a}_{i}\frac{\der}{\der q^{j}}.
\]
Considering the $m\times m$ matrix $(C^{ab})=(\mu_i^a\,{\mathcal G}^{ij}\,\mu_j^b)$ (which is symmetric and regular since ${\mathcal G}$ is a Riemannian metric), we obtain the local description of $\mathcal{Q}$:
\[
\mathcal{Q}=C_{ab}Z^{a}\otimes\mu^{b}=C_{ab}{\mathcal G}^{ij}\mu_{i}^{a}\mu_{k}^{b}\frac{\der}{\der q^{j}}\otimes dq^{k},
\]
and $\mathcal{P}=\operatorname{Id}_{TQ}-\mathcal{Q}$. 
Finally, by using these projectors we may rewrite the equation of motion as follows. A curve $c\lp t\rp$ is a motion of the nonholonomic system if it satisfies the constraints, i.e., $\dot c\lp t\rp\in\mathcal{D}_{c\lp t\rp}$, and, in addition, the ``projected equation of motion''
\begin{equation}\label{NHProjectedEq}
\mathcal{P}\lp\nabla_{\dot c\lp t\rp}\dot c\lp t\rp\rp=-\mathcal{P}\lp\mbox{grad}\,V\lp c\lp t\rp\rp\rp
\end{equation}
is fulfilled.

Summarizing, we have obtained the dynamics of the nonholonomic system (\ref{NHProjectedEq}) applying the projector $\mathcal{P}$ to the unconstrained equations of motion (\ref{EqsRiemanLibre}).

\section{Discrete Mechanics and Variational Integrators}
\label{section2}

 Variational integrators are a kind of geometric integrators for the Euler-Lagrange equations which retain  their variational character and also, as a consequence,  some of main  geometric properties of the continuous
system, such as symplecticity and momentum conservation (see
\cite{Hairer,MarsdenWest,MoVe,Veselov}).
  In the following we will summarize the main features of this type
of geometric
 integrators.  A \textbf{discrete Lagrangian} is a map
$L_d:Q\times Q\rightarrow\R$, which may be considered as
an approximation of the action integral defined by a continuous  Lagrangian $L\colon TQ\to
\R$, that is,
$
L_d(q_0, q_1)\approx \int^h_0 L(q(t), \dot{q}(t))\; dt,
$
where $q(t)$ is a solution of the Euler-Lagrange equations for $L$ joining $q(0)=q_0$ and $q(h)=q_1$ for small enough $h>0$.

 Define the \textbf{action sum} $S_d\colon Q^{N+1}\to
\R$  corresponding to the Lagrangian $L_d$ by
$
{S_d}=\sum_{k=1}^{N}  L_d(q_{k-1}, q_{k}),
$
where $q_k\in Q$ for $0\leq k\leq N$, where $N$ is the number of steps. The discrete variational
principle   states that the solutions of the discrete system
determined by $L_d$ must extremize the action sum given fixed
endpoints $q_0$ and $q_N$. By extremizing ${S_d}$ over $q_k$,
$1\leq k\leq N-1$, we obtain the system of difference equations
\begin{equation}\label{discreteeq}
 D_1L_d( q_k, q_{k+1})+D_2L_d( q_{k-1}, q_{k})=0.
\end{equation}

These  equations are usually called the  \textbf{discrete
Euler--Lagrange equations}. Under some regularity hypotheses (the
matrix $(D_{12}L_d(q_k, q_{k+1}))$ is regular), it is possible to
define  from (\ref{discreteeq}) a (local) discrete flow $ \Upsilon_{L_d}\colon Q\times
Q\to  Q\times Q$, by $\Upsilon_{L_d}(q_{k-1}, q_k)=(q_k,
q_{k+1})$. Define the  discrete
Legendre transformations associated to  $L_d$ as
\begin{eqnarray*}
\F L_d^-\colon Q\times Q&\to & T^*Q\\
(q_k, q_{k+1})&\longmapsto& (q_k, -D_1 L_d(q_k, q_{k+1})),\\
\F L_d^+\colon  Q\times Q&\to&  T^*Q\\
(q_k, q_{k+1})&\longmapsto& (q_{k+1}, D_2 L_d(q_k, q_{k+1})),
\end{eqnarray*}
and the discrete Poincar{\'e}--Cartan 2-form $\omega_d=(\F L_d^+)^*\omega_Q=(\F L_d^-)^*\omega_Q$,
where $\omega_Q$ is the canonical symplectic form on $T^*Q$. The
discrete algorithm determined by $\Upsilon_{L_d}$ preserves the
symplectic form $\omega_d$, i.e., $\Upsilon_{L_d}^*\omega_d=\omega_d$.
Moreover, if the discrete Lagrangian is invariant under the
diagonal action of a Lie group $G$, then the discrete momentum map
$J_d\colon Q\times Q \to  {\mathfrak g}^*$ defined by
\[ \langle
J_d(q_k, q_{k+1}), \xi\rangle=\langle D_2L_d(q_k, q_{k+1}),
\xi_Q(q_{k+1})\rangle \]
is preserved by the discrete flow.
Therefore, these integrators are symplectic-momentum preserving. Here, $\xi_Q$ denotes the fundamental vector field
determined by $\xi\in {\mathfrak g}$, where ${\mathfrak g}$ is the Lie
algebra of $G$. (See \cite{MarsdenWest} for more details.)

\section{The Geometric Nonholonomic Integrator}\label{section4}

The Geometric Nonholonomic Integrator (GNI in the sequel) and its principal features have been presented in \cite{SDD1,SDD2,SDM}.  As main geometric properties, we can mention that it preserves the nonholonomic constraints, the discrete nonholonomic momentum map in the presence of horizontal symmetries, and the energy of the system under certain symmetry conditions \cite{SDD1}.

\begin{definition}
Consider a discrete Lagrangian $L_{d}:Q\times Q\Flder\R$. The proposed {\rm discrete nonholonomic equations} are
\begin{subequations}\label{GNIint}
\begin{align}
\mathcal{P}_{q_{k}}^*\lp D_{1}L_{d}\lp q_{k},q_{k+1}\rp\rp+\mathcal{P}_{q_{k}}^*\lp D_{2}L_{d}\lp q_{k-1},q_{k}\rp\rp=0,\label{GNIinta}\\
\mathcal{Q}_{q_{k}}^*\lp D_{1}L_{d}\lp q_{k},q_{k+1}\rp\rp-\mathcal{Q}_{q_{k}}^*\lp D_{2}L_{d}\lp q_{k-1},q_{k}\rp\rp=0,\label{GNIintb}
\end{align}
\end{subequations}
which define the GNI integrator.
\end{definition}
The projectors $\mathcal{P}$, $\mathcal{Q}$ are defined in the previous sections, while the subscript $q_{k}$ emphasizes that the projections take place in the fiber $T^*_{q_{k}}Q$. The first equation is just the projection of the discrete Euler-Lagrange equation to the constraint distribution $\mathcal{D}$, while the second one can be interpreted as an elastic impact of the system against $\mathcal{D}$ (see \cite{DIbort}). Note that since $\mathcal{P}$ and $\mathcal{Q}$ are orthogonal and complementary, \eqref{GNIint} is equivalent to
\begin{equation}\label{Comb}
D_{1}L_{d}(q_{k},q_{k+1})+(\mathcal{P}^*-\mathcal{Q}^*)\,D_{2}L_{d}(q_{k-1},q_{k})=0.
\end{equation}
From these equations we see that the system defines a unique discrete evolution operator if and only if the matrix $\lp D_{12}L_{d}\rp$ is regular, that is, the discrete Lagrangian is regular. Locally, equations \eqref{GNIint} can be written as
\begin{subequations}\label{GNIloc}
\begin{align}
D_{1}L_{d}\lp q_{k},q_{k+1}\rp+D_{2}L_{d}\lp q_{k-1},q_{k}\rp&=\lp\lambda_{k}\rp_{b}\mu^{b}(q_k),\label{GNIloca}\\
{\mathcal G}^{ij}\lp q_{k}\rp\mu_{i}^{a}\lp q_{k}\rp\lp\frac{\der L_{d}}{\der x^{j}}\lp q_{k},q_{k+1}\rp -\frac{\der L_{d}}{\der y^{j}}\lp q_{k-1},q_{k}\rp\rp&=0\label{GNIlocb}.
\end{align}
\end{subequations}
Using the discrete Legendre transformations defined above, let us define the pre- and post-momenta, which are covectors at $q_{k}$, by
\begin{align*}
p^{+}_{k-1,k}=&p^{+}\lp q_{k-1},q_{k}\rp=\F L_{d}^+\lp q_{k-1},q_{k}\rp=D_{2}L_{d}\lp q_{k-1},q_{k}\rp\\
p^{-}_{k,k+1}=&p^{-}\lp q_{k},q_{k+1}\rp=\F L_{d}^-\lp q_{k},q_{k+1}\rp=-D_{1}L_{d}\lp q_{k},q_{k+1}\rp.
\end{align*}
Then, the second GNI equation (\ref{GNIlocb}) can be rewritten as follows:
\[
{\mathcal G}^{ij}\lp q_{k}\rp\mu_{i}^{a}\lp q_{k}\rp\lp\frac{( p_{k,k+1}^{-})_{j}+( p_{k-1,k}^{+})_{j}}{2}\rp=0,
\]
which means that the average of pre- and post-momenta satisfies the constraints. In this sense the proposed numerical method preserves exactly the nonholonomic constraints. Besides this preservation property, the GNI has other interesting geometric features like the preservation of energy when the configuration manifold is a Lie group with a Lagrangian defined by a bi-invariant metric, with an arbitrary distribution $\mathcal{D}$ and a discrete Lagrangian that is left-invariant (see \cite{SDD1} for further details).

\section{GNI extensions of symplectic-Euler methods}\label{SymEuExt}
Let us consider the tangent $TQ$ and cotangent $T^*Q$ bundles of the configuration manifold $Q=\R^n$ and its local coordinates, $(q,\dot q)$ and $(q,p)$ respectively. Moreover, let us consider the mechanical Lagrangian $L(q,\dot q)=\frac{1}{2}\,\dot q^T\,M\,\dot q-V(q)$, where $M$ is a $n\times n$ constant regular matrix and $V:Q\Flder\R$ is the potential function. On the other hand, the function $H(q,p)=\frac{1}{2}\,p^T\,M^{-1}\,p+V(q)$ is its Hamiltonian counterpart.

It is well known that the explicit and implicit Euler methods (which we will denote Euler A and Euler B respectively)
\[
\begin{array}{c@{\hskip 1.5cm}c}
\mbox{Euler A} & \mbox{Euler B}\\
q_{k+1}=q_{k}+hM^{-1}p_{k} & q_{k+1}=q_{k}+hM^{-1}p_{k+1}\\
p_{k+1}=p_{k}-h\displaystyle\frac{\der V}{\der q}\lp q_{k}\rp & p_{k+1}=p_{k}-h\displaystyle\frac{\der V}{\der q}\lp q_{k+1}\rp
\end{array}
\]
are symplectic and of order one (see \cite{Hairer}). As variational integrators (see \cite{MarsdenWest}) they correspond to the following discrete Lagrangians:
\begin{small}
\begin{equation}\label{LagEuler}
L_{d}^{A}(q_{k},q_{k+1})=hL\left(q_{k},\frac{q_{k+1}-q_{k}}{h}\right),\,\,\,L_{d}^{B}(q_{k},q_{k+1})=hL\left(q_{k+1},\frac{q_{k+1}-q_{k}}{h}\right).
\end{equation}
\end{small}%
Applying the GNI equations \eqref{GNIloc} to the Lagrangians in (\ref{LagEuler}) we obtain the following numerical schemes:
\vspace{0.3cm}

$\bullet$ {\bf Euler A}:
\begin{subequations}\label{EuALag}
\begin{align}
q_{k+1}-2q_{k}+q_{k-1}&=-h\cua M^{-1}\lp V_{q}(q_{k})+\mu^{T}\lp q_{k}\rp\tilde\lambda_{k}\rp,\label{EuALaga}\\
0&=\mu(q_{k})\lp\frac{q_{k+1}-q_{k-1}}{2h}+\frac{h}{2}M^{-1}V_q(q_k)\rp.\label{EuALagb}
\end{align}
\end{subequations}

$\bullet$ {\bf Euler B}:
\begin{subequations}\label{EuBLag}
\begin{align}
q_{k+1}-2q_{k}+q_{k-1}&=-h\cua M^{-1}\lp V_{q}(q_{k})+\mu^{T}\lp q_{k}\rp\tilde\lambda_{k}\rp,\label{EuBLaga}\\
0&=\mu(q_{k})\lp\frac{q_{k+1}-q_{k-1}}{2h}-\frac{h}{2}M^{-1}V_q(q_k)\rp\label{EuBLagb},
\end{align}
\end{subequations}
where $\tilde\lambda_k=\lambda_k/h$ and $V_q=\der V/\der q$. Observe that the only difference between the two methods lies in the sign between parentheses in \eqref{EuALagb} and \eqref{EuBLagb}. By introducing the momentum quantities $\tilde p_k=M(q_{k+1}-q_{k-1})/2h$ and $p_{k+1/2}=M(q_{k+1}-q_k)/h$, we can rewrite equations \eqref{EuALag} and \eqref{EuBLag} as follows.
\vspace{0.3cm}

$\bullet$ {\bf Euler A}:
\begin{subequations}\label{EuAHam}
\begin{align}
&p_{k+1/2}=\tilde p_{k}-\frac{h}{2}\lp V_{q}(q_{k})+\mu^{T}(q_{k})\tilde\lambda_{k}\rp,\label{EuAHama}\\
&q_{k+1}=q_{k}+hM^{-1}p_{k+1/2},\label{EuAHamb}\\
&\mu(q_{k}) M^{-1}\lp\tilde p_{k}+\frac{h}{2}V_q(q_k)\rp=0,\label{EuAHamc}\\
&\tilde p_{k+1}=p_{k+1/2}-\frac{h}{2}\lp V_{q}(q_{k+1})+\mu^{T}(q_{k+1})\tilde\lambda_{k+1}\rp,\label{EuAHamd}\\
&\mu(q_{k+1}) M^{-1}\lp\tilde p_{k+1}+\frac{h}{2}V_q(q_{k+1})\rp=0.\label{EuAHame}
\end{align}
\end{subequations}

$\bullet$ {\bf Euler B}:
\begin{subequations}\label{EuBHam}
\begin{align}
&p_{k+1/2}=\tilde p_{k}-\frac{h}{2}\lp V_{q}(q_{k})+\mu^{T}(q_{k})\tilde\lambda_{k}\rp,\label{EuBHama}\\
&q_{k+1}=q_{k}+hM^{-1}p_{k+1/2},\label{EuBHamb}\\
&\mu(q_{k}) M^{-1}\lp\tilde p_{k}-\frac{h}{2}V_q(q_k)\rp=0,\label{EuBHamc}\\
&\tilde p_{k+1}=p_{k+1/2}-\frac{h}{2}\lp V_{q}(q_{k+1})+\mu^{T}(q_{k+1})\tilde\lambda_{k+1}\rp,\label{EuBHamd}\\
&\mu(q_{k+1}) M^{-1}\lp\tilde p_{k+1}-\frac{h}{2}V_q(q_{k+1})\rp=0.\label{EuBHame}
\end{align}
\end{subequations}
These numerical schemes provide values at step $k+1$ through an intermediate momentum step $k+1/2$, i.e.,
\[(q_k,\tilde p_k,\tilde\lambda_k)\Flder(q_{k+1},p_{k+1/2},\tilde\lambda_k)\Flder(q_{k+1},\tilde p_{k+1},\tilde\lambda_{k+1}). \]
We recognize in \eqref{EuAHamc}, \eqref{EuAHame} and \eqref{EuBHamc}, \eqref{EuBHame} a Hamiltonian version for the discretization of the nonholonomic constraints \eqref{EuALagb} and \eqref{EuBLagb} (Lagrangian version). These constraints are provided by the GNI equations \eqref{GNIintb} or \eqref{GNIlocb}.

\begin{remark}
\label{remark1}
Method \eqref{EuALag} (and the corresponding B version) clearly resembles the extension of the SHAKE method  (see \cite{SHAKE}) proposed by R. McLachlan and M. Perlmutter \cite{McLPerl} as a reversible method for nonholonomic systems {\it not based} on the discrete Lagrange-d'Alembert principle, namely
\begin{align*}
q_{k+1}-2q_{k}+q_{k+1}&=-h^2M^{-1}\lp V_q(q_k)+\mu^T(q_k)\tilde\lambda_k\rp,\\
0&=\mu(q_k)\lp\frac{q_{k+1}-q_{k-1}}{2h}\rp.
\end{align*}
At the same time, the SHAKE method is an extension of the classical St\"ormer-Verlet method in the presence of holonomic constraints.  The RATTLE method is algebraically equivalent to SHAKE \cite{EqAlge}. Its nonholonomic extension, introduced for the first time in \cite{McLPerl}, that is
\begin{subequations}\label{nonholonomicRATTLE}
\begin{align}
&p_{k+1/2}=\tilde p_{k}-\frac{h}{2}\lp V_{q}(q_{k})+\mu^{T}(q_{k})\tilde\lambda_{k}\rp,\\
&q_{k+1}=q_{k}+hM^{-1}p_{k+1/2},\\
&\mu(q_{k}) M^{-1}\tilde p_{k}=0,\\
&\tilde p_{k+1}=p_{k+1/2}-\frac{h}{2}\lp V_{q}(q_{k+1})+\mu^{T}(q_{k+1})\tilde\lambda_{k+1}\rp,\\
&\mu(q_{k+1}) M^{-1}\tilde p_{k+1}=0,
\end{align}
\end{subequations}
(see \cite{SDD1}) clearly resembles \eqref{EuAHam}.

As shown in \cite{SDD1}, the nonholonomic SHAKE extension can be obtained by applying the GNI equations to the discrete Lagrangian
\begin{equation}\label{Verlet}
L_d(q_k,q_{k+1})= \frac{h}{2}L\left(q_k,\frac{q_{k+1}-q_k}{h}\right)+\frac{h}{2}L\left(q_{k+1},\frac{q_{k+1}-q_k}{h}\right),
\end{equation}
which also provides the St\"ormer-Verlet method in the variational integrators sense. Moreover, as shown in \cite{SDD2}, the nonholonomic RATTLE method \eqref{nonholonomicRATTLE} is globally second-order convergent.
\end{remark}

\begin{theorem}\label{Teo1}
The nonholonomic extension of the Euler A (B) method is globally first-order convergent.
\end{theorem}
It will be useful in the following proof to give a Hamiltonian version of \eqref{LdAeqs} when $H(q,p)=\frac{1}{2}\,p^T M^{-1}p+V(q)$, namely
\begin{align*}
\dot q&=M^{-1}p,\\
\dot p&=-V_q(q)-\mu^T(q)\lambda,\\
\mu(q)\,M^{-1}p&=0.
\end{align*}
Since the constraints are satisfied along the solutions, we can differentiate them w.r.t.\ time in order to obtain the actual values of the Lagrange multipliers, i.e.
\[
\lambda=\mathcal{C}^{-1}\lp\mu_q[M^{-1}p,M^{-1}p]-\mu M^{-1}V_q\rp,
\]
where $\mathcal{C}(q)=\mu(q) M^{-1}\mu^{T}(q)$ is a regular matrix and $\mu_{q}[M^{-1}p,M^{-1}p]$ is the $m\times 1$ matrix
$\frac{\der\mu^{\alpha}_{i}}{\der q^{j}}\lp M^{-1}\rp^{jj^{\prime}}p_{j^{\prime}}\lp  M^{-1}\rp^{ii^{\prime}}p_{i^{\prime}}$. Taking this into account, the Hamiltonian nonholonomic system becomes
\begin{subequations}\label{H}
\begin{align}
\dot q&=M^{-1}p,\label{Ha}\\
\dot p&=-V_{q}-\mu^{T}\mathcal{C}^{-1}\lp\mu_{q}[M^{-1}p,M^{-1}p]-\mu M^{-1}V_q\rp\label{Hb},
\end{align}
\end{subequations}
with initial condition satisfying $\mu(q)M^{-1}p=0$.
\begin{proof} [Proof of Theorem \ref{Teo1}]

We present the proof for the Euler A method, the corresponding proof for Euler B is analogous.

Consider the unconstrained problem
\begin{align*}
\dot q&=M^{-1}p,\\
\dot p&=\phi\lp q,p\rp,
\end{align*}
with a smooth enough function $\phi:\R^{2n}\Flder\R$. These equations can be discretized by
\begin{subequations}\label{P1}
\begin{align}
q_{k+1}&=q_{k}+hM^{-1} p_{k+1/2},\label{P1a}\\
p_{k+1/2}&=p_{k-1/2}+h\phi\lp q_{k},p_{k+1/2}\rp,\label{P1b}
\end{align}
\end{subequations}
which is a globally first-order convergent method, using standard arguments of Taylor expansions. Therefore, taking into account equations \eqref{H}, from \eqref{P1} we deduce the following first-order method for the nonholonomic system
\begin{subequations}\label{P2}
\begin{align}
q_{k+1}&=q_{k}+hM^{-1}p_{k+1/2},\label{P2a}\\
p_{k+1/2}&=p_{k-1/2}-hV_{q}\lp q_{k}\rp+h\mu^{T}\lp q_{k}\rp\mathcal{C}^{-1}\lp q_{k}\rp\mu\lp q_{k}\rp M^{-1}V_{q}\lp q_{k}\rp\nonumber\\
&\quad -h\mu^{T}\lp q_{k}\rp\mathcal{C}^{-1}\lp q_{k}\rp\mu_{q}[M^{-1}p_{k+1/2},M^{-1}p_{k+1/2}].\label{P2b}
\end{align}
\end{subequations}
The next step is to prove that the nonholonomic Euler A method \eqref{EuAHam} reproduces \eqref{P2}. From equations \eqref{EuAHam} we see that the nonholonomic Euler A method assumes the form
\begin{align*}
q_{k+1}&=q_{k}+hM^{-1}p_{k+1/2},\\
p_{k+1/2}&=p_{k-1/2}-hV_q(q_k)-h\mu^T(q_k)\tilde\lambda_k,\\
0&=\mu(q_k)M^{-1}\lp\frac{p_{k+1/2}+p_{k-1/2}}{2}+\frac{h}{2}V_q(q_k)\rp
\end{align*}
or, after some computations,
\begin{subequations}\label{P3}
\begin{align}
q_{k+1}&=q_{k}+hM^{-1}p_{k+1/2},\label{P3a}\\
p_{k+1/2}&=p_{k-1/2}-hV_{q}(q_{k})-2\mu^{T}(q_{k})\mathcal{C}^{-1}(q_{k})\mu(q_{k}) M^{-1}p_{k-1/2}.\label{P3b}
\end{align}
\end{subequations}
On the other hand we can expand the nonholonomic constraints around $q\lp 0\rp$:
\[
\mu\lp q\lp h\rp\rp\dot q\lp h\rp=\mu\lp q\lp 0\rp\rp\dot q\lp 0\rp+h\mu\lp q\lp 0\rp\rp\ddot q\lp 0\rp+h\mu_{q}[\dot q\lp 0\rp,\dot q\lp 0\rp]+\mathcal{O}\lp h\cua\rp.
\]
Since the constraints are satisfied at $t=0$ and $t=h$, the previous expression becomes
\[
h\mu\lp q\lp 0\rp\rp\ddot q\lp 0\rp=-h\mu_{q}[\dot q\lp 0\rp,\dot q\lp 0\rp]+\mathcal{O}\lp h\cua\rp.
\]
Now, taking standard aproximations for first and second derivatives we deduce that
\begin{multline}
-2\mu(q_k)M^{-1}p_{k-1/2}=-h\mu_q[M^{-1}p_{k+1/2},M^{-1}p_{k+1/2}]\\
+h\mu(q_k)M^{-1}V_q(q_k)+\mathcal{O}(h^2).\label{P4}
\end{multline}
Therefore, substituting \eqref{P4} into \eqref{P3b} we recognize equation \eqref{P2b} up to $\mathcal{O}(h^2)$ terms. Thus, we conclude that the nonholonomic Euler A method \eqref{EuAHam} is first-order convergent.
\end{proof}

\begin{definition}
For a one-step method $F:T^*Q\Flder T^*Q$, the adjoint method $F^*:T^*Q\Flder T^*Q$ is defined by
\[
(F^*)^{h}\circ F^{-h}=\mbox{Id}_{_{T^*Q}}.
\]
\end{definition}

\begin{theorem}\label{adThe}
The nonholonomic extensions of the Euler A and B methods are one another's adjoint.
\end{theorem}

\begin{proof}
We will use a shorthand notation to define both integrators:
\begin{align*}
F_A(q_k,\tilde p_k,\tilde\lambda_k)&=(q_{k+1}^A,\tilde p_{k+1}^A,\tilde\lambda_{k+1}^A)\\
F_B(q_k,\tilde p_k,\tilde\lambda_k)&=(q_{k+1}^B,\tilde p_{k+1}^B,\tilde\lambda_{k+1}^B).
\end{align*}
Equations \eqref{EuAHam} and \eqref{EuBHam} can be rewritten to give a one-step method instead of the leap-frog presented. For instance, for $F_A$,
\begin{subequations}\label{FA}
\begin{align}
q_{k+1}^A&=q_k+hM^{-1}\tilde p_k-\frac{h^2}{2}M^{-1}V_q(q_k)-\frac{h^2}{2}M^{-1}\mu^T(q_k)\tilde\lambda_k,\label{FAa}\\
\tilde p_{k+1}^A&=\tilde p_k-\frac{h}{2}V_q(q_k)-\frac{h}{2}\mu^T(q_k)\tilde\lambda_k-\frac{h}{2}V_q(q_{k+1}^A)-\frac{h}{2}\mu^T(q_{k+1}^A)\tilde\lambda_{k+1}^A,\label{FAb}\\
0&=\mu(q_{k+1}^A)M^{-1}\tilde p_k-\frac{h}{2}\mu(q_{k+1}^A)M^{-1}V_q(q_k)-\frac{h}{2}\mu(q_{k+1}^A)M^{-1}\mu^T(q_k)\tilde\lambda_k\nonumber\\
&\quad -\frac{h}{2}\mu(q_{k+1}^A)M^{-1}\mu^T(q_{k+1}^A)\tilde\lambda_{k+1}^A,\label{FAc}
\end{align}
\end{subequations}
where $\tilde p_{k+1}^A$ and $\tilde\lambda_{k+1}^A$ are implicitly obtained from \eqref{FAb} and \eqref{FAc}. The same occurs for $F_B$:
\begin{subequations}\label{FB}
\begin{align}
q_{k+1}^B&=q_k+hM^{-1}\tilde p_k-\frac{h^2}{2}M^{-1}V_q(q_k)-\frac{h^2}{2}M^{-1}\mu^T(q_k)\tilde\lambda_k,\label{FBa}\\
\tilde p_{k+1}^B&=\tilde p_k-\frac{h}{2}V_q(q_k)-\frac{h}{2}\mu^T(q_k)\tilde\lambda_k-\frac{h}{2}V_q(q_{k+1}^B)-\frac{h}{2}\mu^T(q_{k+1}^B)\tilde\lambda_{k+1}^B,\label{FBb}\\
0&=-\mu(q_k)M^{-1}\tilde p_{k+1}^B
-\frac{h}{2}\mu(q_k)M^{-1}V_q(q_{k+1}^B)
-\frac{h}{2}\mu(q_k)M^{-1}\mu^T(q_k)\tilde\lambda_k
\nonumber\\
&\quad -\frac{h}{2}\mu(q_k)M^{-1}\mu^T(q_{k+1}^B)\tilde\lambda_{k+1}^B
.\label{FBc}
\end{align}
\end{subequations}

The point of the proof is to show that $F_A^h\circ F_B^{-h}(q_k,\tilde p_k,\tilde\lambda_k)=(q_k,\tilde p_k,\tilde\lambda_k)$. In order to do that, we are going to use the notation
\begin{align*}
F_B^{-h}(q_k,\tilde p_k,\tilde\lambda_k)&=(q_{k+1},\tilde p_{k+1},\tilde\lambda_{k+1})=(q_k^{\prime},\tilde p_k^{\prime},\tilde\lambda_k^{\prime}),\\
F_A^{h}(q_k^{\prime},\tilde p_k^{\prime},\tilde\lambda_k^{\prime})&=(q_{k+1}^{\prime},\tilde p_{k+1}^{\prime},\tilde\lambda_{k+1}^{\prime}),
\end{align*}
so we need to show that $(q_{k+1}^{\prime},\tilde p_{k+1}^{\prime},\tilde\lambda_{k+1}^{\prime})=(q_k,\tilde p_k,\tilde\lambda_k)$. After setting the time step to $-h$ and replacing \eqref{FBa} and \eqref{FBb} into \eqref{FAa}, it is easy to check that $q_{k+1}^{\prime}=q_k$. Furthermore, fixing $-h$ again as the time step and taking into account equation \eqref{EuBHame}, from \eqref{FBc} we arrive to
\begin{multline*}
-\frac{h}{2}M^{-1}\mu^T(q_k^{\prime})\tilde\lambda_k^{\prime}=\\
-M^{-1}\tilde p_k-\frac{h}{2}V_q(q_k)-M^{-1}\tilde p_k^{\prime}+\frac{h}{2}V_q(q_k^{\prime})+\frac{h}{2}M^{-1}\mu^T(q_k)\tilde\lambda_k.
\end{multline*}
Replacing this expression into \eqref{FAc}, considering that $q_{k+1}^{\prime}=q_k$ and taking into account $\eqref{EuAHame}$ we find that
\[
\frac{h}{2}\mu(q_k)M^{-1}\mu^T(q_k)\tilde\lambda_k-\frac{h}{2}\mu(q_k)M^{-1}\mu^T(q_k)\tilde\lambda_{k+1}^{\prime}=0,
\]
which means
\[
\tilde\lambda_{k+1}^{\prime}=\tilde\lambda_k
\]
since $\mathcal{C}(q_k)$ is regular. Finally, replacing \eqref{FBb} into \eqref{FAb} we find that $\tilde p_{k+1}^{\prime}=\tilde p_k$.
\end{proof}
\begin{remark}
As shown in \cite{MarsdenWest}, the composition of Hamiltonian discrete flows, in the {\it variational integrators} sense, generated by the discrete Lagrangians \eqref{LagEuler} reproduces the RATTLE algorithm in the free case (that is, not constrained). More concretely, the composition
\[
F_{L_A}^{h/2}\circ F_{L_B}^{h/2}
\]
produces the algorithm
\begin{align*}
p_{k+1/2}&=\tilde p_k-\frac{h}{2}V_q(q_k),\\
q_{k+1}&=q_k+hM^{-1}p_{k+1/2},\\
\tilde p_{k+1}&= p_{k+1/2}-\frac{h}{2}V_q(q_{k+1}).
\end{align*}
Unfortunately, this is no longer true in the nonholonomic case, i.e., one can check that the composition (with time step $h/2$) of methods $\eqref{EuAHam}$ and $\eqref{EuBHam}$ does not reproduce the equations presented in remark \ref{remark1}. However, this composition still generates a second order method since the intermediate steps are first order methods which are each other's adjoint (as we have just proved).
\end{remark}

\section{Affine extension of the GNI}\label{AffS}

We consider in this section the case of affine noholonomic constraints determined by an affine subbundle ${\mathcal A}$ of $TQ$ modeled on a vector subbundle ${\mathcal D}$. We will assume, in the sequel, that there exists a globally defined vector field $Y\in \mathfrak{X}(Q)$ such that $v_q\in {\mathcal A}_q$ if and only if
 $v_q-Y(q)\in {\mathcal D}_q$.
 Therefore, if ${\mathcal D}$ is determined by constraints $\mu^a_i(q)\dot{q}^i=0$, then ${\mathcal A}$ is locally determined by the vanishing of the constraints
 \begin{equation}\label{AC}
\phi^{a}\lp q,\dot q\rp=\mu^{a}_{i}(q)\lp\dot q^{i}-Y^i(q)\rp=0,\hspace{2mm} 1\leq a\leq m.
\end{equation}
where $Y=Y^i\frac{\partial}{\partial q^i}$.

In consequence, the initial data defining our {\it nonholonomic affine system} is denoted by the 4-tuple $(\mathcal{D}, {\mathcal G}, Y, V)$, where $\mathcal{D}$ is the distribution, ${\mathcal G}$ the Riemannian metric, $Y$ the globally defined vector field  and $V$ the potential function. By means of the metric, from $Y$, we can uniquely define a 1-form ${\mathcal G}(Y,\cdot)=\Pi\in\Omega^1(Q)$. Locally, $\Pi={\mathcal G}_{ij}Y^j\, dq^i$.

In terms of momenta the nonholonomic constraints  \eqref{AC} can be rewritten as
\begin{equation}\label{ACHam}
\mu^a_i(q)\,{\mathcal G}^{ij}\lp p_j-\Pi_j(q)\rp=0.
\end{equation}
where $p_i={\mathcal G}_{ij}\dot{q}^j$.
\begin{definition}\label{GNIAff}
Consider a discrete Lagrangian $L_d:Q\times Q\Flder\R$. The proposed discrete equations for affine nonholonomic constraints are
\begin{subequations}\label{GNIintAff}
\begin{align}
\mathcal{P}_{q_{k}}^*\lp D_{1}L_{d}\lp q_{k},q_{k+1}\rp\rp&+\mathcal{P}_{q_{k}}^*\lp D_{2}L_{d}\lp q_{k-1},q_{k}\rp\rp=0,\label{GNIintAffa}\\
\mathcal{Q}_{q_{k}}^*\lp D_{1}L_{d}\lp q_{k},q_{k+1}\rp\rp&-\mathcal{Q}_{q_{k}}^*\lp D_{2}L_{d}\lp q_{k-1},q_{k}\rp\rp+2\mathcal{Q}_{q_{k}}^*\Pi=0,\label{GNIintAffb}
\end{align}
\end{subequations}
which define the {\rm affine extension of the GNI method}.
\end{definition}
As before, $\mathcal{Q}$ and $\mathcal{P}$ are the projectors defined in section \ref{Riemm}. Locally, the method \eqref{GNIintAff} can be written as
\begin{subequations}\label{GNIlocAff}
\begin{align}
D_{1}L_{d}(q_{k},q_{k+1})+D_{2}L_{d}(q_{k-1},q_{k})=\lp\lambda_{k}\rp_{b}\mu^{b}(q_k),\label{GNIlocAffa}\\
{\mathcal G}^{ij}(q_{k})\mu_{i}^{a}(q_{k})\lp\frac{\der L_{d}}{\der x^{j}}(q_{k},q_{k+1})-\frac{\der L_{d}}{\der y^{j}}(q_{k-1},q_{k})+2\Pi_j(q_k)\rp=0\label{GNIlocAffb}.
\end{align}
\end{subequations}
Using the pre- and post-momenta defined in section \ref{section4}, equation \eqref{GNIlocAffb} can be rewritten as
\[
{\mathcal G}^{ij}(q_{k})\mu_{i}^{a}(q_{k})\lp\frac{\lp p_{k,k+1}^{-}\rp_{j}+\lp p_{k-1,k}^{+}\rp_{j}}{2}-{\Pi_j(q_k)}\rp=0,
\]
which corresponds to the discretization of the affine constraints \eqref{ACHam} on the Hamiltonian side.

\subsection{A theoretical result: nonholonomic SHAKE and RATTLE extensions for affine systems} Let us consider again the mechanical Lagrangian $L(q,\dot q)=\frac{1}{2}\dot q^T\,M\,\dot q-V(q)$ and the discretization presented in \eqref{Verlet}. Applying the affine GNI equations \eqref{GNIlocAff} we obtain:
\begin{subequations}\label{Afin}
\begin{align}
q_{k+1}-2q_{k}+q_{k-1}&=-h\cua M^{-1}\lp V_{q}(q_{k})+\mu^{T}\lp q_{k}\rp\tilde\lambda_{k}\rp,\label{Afina}\\
0&=\mu(q_{k})\lp\frac{q_{k+1}-q_{k-1}}{2h}-Y(q_k)\rp,\label{Afinb}
\end{align}
\end{subequations}
which can be regarded as the extension of the SHAKE algorithm to affine nonholonomic systems. Denoting $\tilde p_k=M(q_{k+1}-q_{k-1})/2h$ and $p_{k+1/2}=M(q_{k+1}-q_k)/h$, from \eqref{Afin} we arrive to
\begin{align*}
&p_{k+1/2}=\tilde p_{k}-\frac{h}{2}\lp V_{q}(q_{k})+\mu^{T}(q_{k})\tilde\lambda_{k}\rp,\\
&q_{k+1}=q_{k}+hM^{-1}p_{k+1/2},\\
&\mu(q_{k}) M^{-1}\lp\tilde p_{k}-\Pi(q_k)\rp=0,\\
&\tilde p_{k+1}=p_{k+1/2}-\frac{h}{2}\lp V_{q}(q_{k+1})+\mu^{T}(q_{k+1})\tilde\lambda_{k+1}\rp,\\
&\mu(q_{k+1}) M^{-1}\lp\tilde p_{k+1}-\Pi(q_{k+1})\rp=0,
\end{align*}
which can be regarded as the extension of the RATTLE algorithm to affine nonholonomic systems.

\section{Reduced systems}
\label{RS}
In this section we are going to consider configuration spaces of the form $Q=M\times G$, where $M$ is an $n$-dimensional differentiable manifold and $G$ is an $m$-dimensional Lie group ($\al$ will be its corresponding Lie algebra). Therefore, there exists a global  canonical splitting between variables describing the position and variables describing the orientation of the mechanical system. Then, we distinguish the pose coordinates $g\in G$, and the variables describing the internal shape of the system, that is $x\in M$ (in consequence $(x,\dot x)\in TM$).  It is clear that $Q=M\times G$ is the total space of a trivial principal $G$-bundle over $M$, where the bundle projection $\phi: Q\Flder M$ is just the canonical projection onto the first factor. We may consider the corresponding reduced tangent space $E=TQ/G$ over $M$. Identifying $TG$ with $G\times \al$ by using left  translations, $E=TQ/G$ is isomorphic to the product manifold $TM\times \al$ and the vector bundle projection is $\tau_M\circ pr_1$, where $pr_1:TM\times \al\Flder TM$ and $\tau_M:TM\Flder M$ are the canonical projections.

\subsection{The case of linear constraints}\label{wert}

Now suppose that $({\mathcal G},  {\mathcal D}, V)$ is a standard mechanical nonholonomic system on $TQ$ such that all the ingredients are $G$-invariant. In other words, for all $x\in M$ and $g\in G$, 
\begin{itemize}
\item ${\mathcal G}_{(x,g)}((X_x, g\xi), (Y_x, g\eta))={\mathcal G}_{(x, e)}((X_x, \xi), (Y_x, \eta))$ for all $X_x, Y_x\in T_xM$, $\xi, \eta\in {\mathfrak g}$;
\item $(X_x, \xi)\in {\mathcal D}_{(x,e)}$ implies $(X_x, g\xi)\in {\mathcal D}_{(x,g)}$;
\item $V(x, g)=V(x,e)\equiv \tilde{V}(x)$.
\end{itemize}
Therefore, we obtain a new triple $(\tilde{{\mathcal G}},  \widetilde{{\mathcal D}}, \tilde{V})$ on $TM\times {\mathfrak g}$ where
$\tilde{{\mathcal G}}: (TM\times {\mathfrak g})\times (TM\times {\mathfrak g})\longrightarrow \R$ is a bundle metric,  $\widetilde{{\mathcal D}}$ is a vector subbundle of $TM\times {\mathfrak g}\rightarrow M$ and $\tilde{V}: M\rightarrow \R$ is the reduced potential.
With all these ingredients it is possible to write the reduced nonholonomic equations or {\em nonholonomic Lagrange-Poincar\'e equations} (see \cite{BlKrMaMu1996,CoLeMaMar} for all the details, also for the non-trivial case).

Our objective is to find a discrete version of the GNI for the nonholonomic Lagrange-Poincar\'e equations.
As in the previous sections,  we can split the total space $E$ as $E=\tilde\dist\oplus\tilde\dist^{\perp}$, using this time the fibered metric $\tilde{{\mathcal G}}$, and consider the corresponding projectors $\pe:E\Flder\tilde\dist$, $\qu:E\Flder\tilde\dist^{\perp}$. In order to write the discrete nonholonomic equations, it is necessary to set a discrete Lagrangian $L_d: M\times M\times G\Flder\R$,  and the discrete Legendre transforms. Namely (see \cite{MMM}):
\begin{small}
\begin{equation}\label{RedLegTr}
\begin{split}
\F L_d^-\colon M\times M\times G&\to  T^*M\times\dal\\
(x_k, x_{k+1},g_k)&\longmapsto (x_k, -D_1 L_d(x_k, x_{k+1},g_k),r^*_{g_k}D_3L_d(x_k, x_{k+1},g_k)),\\
\F L_d^+\colon  M\times M\times G&\to  T^*M\times\dal\\
(x_k, x_{k+1},g_k)&\longmapsto (x_{k+1}, D_2 L_d(x_k, x_{k+1},g_k),l^*_{g_k}D_3L_d(x_k, x_{k+1},g_k))\; .
\end{split}
\end{equation}
\end{small}%
\begin{definition}\label{GNIRed}
Consider the discrete Legendre transforms defined in \eqref{RedLegTr}. The proposed discrete equations are
\begin{subequations}\label{GNIred}
\begin{align}
&\mathcal{P}_{x_k}^*\lp \F L_d^-(x_k,x_{k+1},g_k)\rp-\mathcal{P}_{x_{k}}^*\lp \F L_d^+(x_{k-1},x_{k},g_{k-1})\rp=0,\label{GNIreda}\\
&\mathcal{Q}_{x_k}^*\lp \F L_d^-(x_k,x_{k+1},g_k)\rp+\mathcal{Q}_{x_{k}}^*\lp \F L_d^+(x_{k-1},x_{k},g_{k-1})\rp=0,\label{GNIredb}
\end{align}
\end{subequations}
which define the {\rm reduced GNI equations}. The subscript $x_k$ emphasizes the fact that the projections take place in the fiber over $x_k$.
\end{definition}
To understand why \eqref{GNIredb} represents a discretization of the nonholonomic constraints, we will work in local coordinates. Take now  local coordinates $(x^i)$ on $M$ and a local basis of sections $\{\tilde{e}_\alpha, \tilde{e}_{a}\}$ of $\Gamma(TM\times\al)$ adapted to the decomposition $\tilde\dist\oplus\tilde\dist^{\perp}$, that is $\tilde{e}_\alpha(x)\in \widetilde{{\mathcal D}}_x$ and $\tilde{e}_a(x)\in \widetilde{{\mathcal D}}^\perp_x$, for all $x\in M$.
We have that
\[
\tilde{{\mathcal G}}(\tilde{e}_\alpha, \tilde{e}_\beta)=\tilde{{\mathcal G}}_{\alpha\beta}, \quad \tilde{{\mathcal G}}(\tilde{e}_a, \tilde{e}_\beta)=0, \quad
\tilde{{\mathcal G}}(\tilde{e}_a, \tilde{e}_b)=\tilde{{\mathcal G}}_{ab}\; .
\]

Consider the induced adapted local coordinates $(x^i,y^\alpha, y^{a})$ for $\Gamma(TM\times\al)$. The nonholonomic constraints are represented by $y^{a}=0$ on $E$. Taking the dual basis $\{\tilde{e}^\alpha, \tilde{e}^{a}\}$ of  $\Gamma(T^*M\times\al^*)$, we have induced local coordinates $(x^i, p_{\alpha}, p_a)$ on the Hamiltonian side,
and now the  nonholonomic constraints are represented by $p_a=0$.

On the other hand, in this basis the projector ${\mathcal Q}$ has the expression
\begin{equation}\label{QQ}
{\mathcal Q}=\tilde{e}^a\otimes\tilde{e}_a .
\end{equation}
Define the pre- and post-momenta by
\begin{align*}
p^-_{x_k}&=\F L_d^-(x_k,x_{k+1},g_k)\in T_{x_k}^*M\times\dal,\\
p^+_{x_k}&=\F L_d^+(x_{k-1},x_{k},g_{k-1})\in T_{x_k}^*M\times\dal.
\end{align*}
From equation \eqref{GNIredb} we obtain
\begin{equation}\label{awqe}
\qu^*_{x_k}\lp\frac{p^+_{x_k} + p^-_{x_k}}{2}\rp=0\; .
\end{equation}
If  $p^+_{x_k}=p^+_{\alpha} e^{\alpha}(x_k)+p^+_{a} e^{a}(x_k)$ and  $p^-_{x_k}=p^-_{\alpha} e^{\alpha}(x_k)+p^-_{a} e^{a}(x_k)$, then condition (\ref{awqe}) is expressed using \eqref{QQ} as
\[
 \frac{p_a^++p_a^-}{2}=0,
\]
which means that the average of post and pre-momenta satisfies the nonholonomic constraints  written on the Hamiltonian side.

\subsection{A theoretical result: RATTLE algorithm for reduced spaces}
\label{RS1}

Let us consider $M=\R^n$. Thus, $Q=\R^n\times G$ and $E=TQ/G\cong T\R^n\times\al$. Take a basis $\{E_s\}$ of the Lie algebra $\al$, and consider the following global basis  of $\Gamma(T\R^n\times \al)$
\[ \left\{ \left( \frac{\partial}{\partial x^i},0 \right),\left( 0,E_s \right) \right\}. \]
Therefore, its dual basis is
\[  \left\{ \left( dx^i,0 \right),\left( 0,E^s \right) \right\}. \]
Writing $dx^i \equiv\left( dx^i,0 \right)$ and $E^s\equiv\left( 0,E^s \right)$ for short, the bundle metric $\tilde{\mathcal{G}}$ is written in this basis of sections as
\[ \tilde{\mathcal{G}}=\tilde{\mathcal{G}}_{ij}dx^i\otimes dx^j+\tilde{\mathcal{G}}_{it}dx^i\otimes E^t+\tilde{\mathcal{G}}_{sj}E^s\otimes dx^j+\tilde{\mathcal{G}}_{st}E^s\otimes E^t, \]
Assume that, in this expression, the coefficients of the bundle metric are symmetric and constant, that is, they  do not depend on the base coordinates $x$.
For instance, a typical example would be the kinetic energy bundle metric corresponding to the Lagrangian
\[
L(x, \dot{x}, \xi)=\frac{1}{2}\dot{x}^T M\dot{x} + \frac{1}{2}\langle \xi, \I \xi\rangle
\]
where $M$ is a regular symmetric matrix and $\I: \al \rightarrow \al^*$ is a symmetric positive definite inertia operator.

 Consider the discrete  Lagrangian $L_d\colon \R^n\times \R^n\times G\to \R$ defined by
\begin{multline*} L_d(x_k,x_{k+1},g_k)=\frac{h}{2}\tilde{\mathcal{G}}_{ij}\left(\frac{x_{k+1}^i-x_k^i}{h}\right)\left(\frac{x_{k+1}^j-x_k^j}{h}\right)+
h\tilde{\mathcal{G}}_{it}\left(\frac{x_{k+1}^i-x_k^i}{h}\right)\frac{\left( \tau^{-1}(g_k) \right)^t}{h}\\
+\frac{h}{2}\tilde{\mathcal{G}}_{st}\frac{\left( \tau^{-1}(g_k) \right)^s}{h}\frac{\left( \tau^{-1}(g_k) \right)^t}{h}
-\frac{h}{2}\left( V(x_k)+V(x_{k+1}) \right),
\end{multline*}
where $\tau:\al\Flder G$ is a {\it retraction map}, which is an analytic local diffeomorphism which maps a neighborhood of $0\in\al$ onto a neighborhood of the neutral element $e\in G$ (see Appendix).
Observe that $\tau^{-1}(g_k)\in \al$ and $\tau^{-1}(g_k)=\left( \tau^{-1}(g_k) \right)^s E_s$.

Additionally, we have the vector subbundle  $\widetilde{{\mathcal D}}$ of $T\R^n\times \al$ prescribing the nonholonomic constraints. Write
$\widetilde{{\mathcal D}}^\circ=\spanop\{\mu_i^ad x^i+\eta^a_s E^s\}$. Equation \eqref{GNIreda} of the GNI method is clearly equivalent to
\[
\F L_d^-(x_k,x_{k+1},g_k)-\F L_d^+(x_{k-1},x_{k},g_{k-1})\in\widetilde{\dist}^{\circ}(x_k),
\]
which in this case splits into
\begin{subequations}\label{Red}
\begin{align}
\frac{1}{h}\tilde{\met}_{ij}(x_{k+1}^j-2x_k^j+x_{k-1}^j)&+\frac{1}{h}\tilde{\met}_{it}\lp(\tau^{-1}(g_{k}))^t
-(\tau^{-1}(g_{k-1})^t)\rp\nonumber\\
+hV_{x^i}(x_k)&=-\lambda_{a,k}\,\mu_i^a(x_k),\label{Reda}\\\nonumber\\
\ell^*_{g_{k-1}}D_3L_d(x_{k-1},x_k,g_{k-1})&-r_{g_{k}}^*D_3L_d(x_{k},x_{k+1},g_{k})=\lambda_{a,k}\,\eta_s^a(x_k)\,E^s,\label{Redb}
\end{align}
\end{subequations}
where $V_{x^i}$ stands for $\der V/\der x^i$, and $\lambda_{a,k}$ are the Lagrange multipliers which might vary in each step.

Equation \eqref{Redb} can be rewritten taking into account the {\it right trivialized tangent retraction map} $\mbox{d}\tau_{\xi}$ for $\xi\in\al$, defined as
\begin{equation}\label{RTT}
\mbox{d}\tau_{\xi}
=
T_{\tau(\xi)}r_{\tau(\xi)^{-1}}
\circ
T_{\xi}\tau\colon\al \to \al,
\end{equation}
where $T_\xi\tau\colon\ T_{\xi}\al\equiv\al \to T_{\tau(\xi)}G$, and its inverse $\mbox{d}\tau_{\xi}^{-1}$ (see also Definition \ref{Retr_again}).

Define the \emph{retracted discrete Lagrangian} $l_d\colon \R^n\times \R^n \times \al \to \R$ as
$l_d(x_k,x_{k+1},\sigma_k)=L_d(x_k,x_{k+1},\tau(\sigma_k))$. For example, for the discrete Lagrangian $L_d$ defined above,
\begin{multline*} l_d(x_k,x_{k+1},\sigma_k)=\frac{h}{2}\tilde{\mathcal{G}}_{ij}\left(\frac{x_{k+1}^i-x_k^i}{h}\right)\left(\frac{x_{k+1}^j-x_k^j}{h}\right)+
h\tilde{\mathcal{G}}_{it}\left(\frac{x_{k+1}^i-x_k^i}{h}\right)\frac{\sigma_k ^t}{h}\\
+\frac{h}{2}\tilde{\mathcal{G}}_{st}\frac{\sigma_k ^s}{h}\frac{\sigma_k ^t}{h}
-\frac{h}{2}\left( V(x_k)+V(x_{k+1}) \right).
\end{multline*}
Note that $\sigma_k/h$ plays the role of a velocity in the Lie algebra direction, so $\sigma_k$ represents a small change in the pose variables after time $h$. In this sense, $\sigma_k$ is analogous to the pair $(x_k,x_{k+1})$. One has
\[
D_3l_d(x_k,x_{k+1},\sigma_{k})=D_3L_d(x_k,x_{k+1},\tau(\sigma_{k}))\circ T_{\sigma_k}\tau.
\]
Using lemma \ref{lemaPr_again} and definition \ref{Retr_again} in the Appendix, one can compute
\begin{align*}
(\mbox{d}\tau^{-1}_{-\sigma_{k}})^*&D_3l_d(x_{k},x_{k+1},\sigma_{k})
=D_3L_d(x_k,x_{k+1},\tau(\sigma_{k}))\circ T_{\sigma_k}\tau\circ\mbox{d}\tau^{-1}_{-\sigma_{k}}=\\
&=D_3L_d(x_k,x_{k+1},\tau(\sigma_{k}))\circ T_{\sigma_k}\tau\circ \mbox{d}\tau^{-1}_{\sigma_{k}}\circ \Ad_{\tau(\sigma_k)}=\\
&=D_3L_d(x_k,x_{k+1},\tau(\sigma_{k}))\circ T_{\sigma_k}\tau\circ 
(T_{\sigma_k}\tau) ^{-1}\circ T_er_{\tau(\sigma_k)}
\circ \Ad_{\tau(\sigma_k)}=\\
&=D_3L_d(x_k,x_{k+1},\tau(\sigma_{k}))\circ T_e\ell_{\tau(\sigma_k)}
\end{align*}
and
\begin{align*}
(\mbox{d}\tau^{-1}_{\sigma_{k}})^*&D_3l_d(x_{k},x_{k+1},\sigma_{k})
=D_3L_d(x_k,x_{k+1},\tau(\sigma_{k}))\circ T_{\sigma_k}\tau\circ\mbox{d}\tau^{-1}_{\sigma_{k}}=\\
&=D_3L_d(x_k,x_{k+1},\tau(\sigma_{k}))\circ T_{\sigma_k}\tau\circ 
(T_{\sigma_k}\tau) ^{-1}\circ T_er_{\tau(\sigma_k)}\\
&=D_3L_d(x_k,x_{k+1},\tau(\sigma_{k}))\circ T_er_{\tau(\sigma_k)}.
\end{align*}
Therefore, setting $g_k=\tau(\sigma_k)$ and $\sigma_k=h\xi_k$, equation \eqref{Redb} becomes
\begin{equation}\label{Redd_2}
\begin{split}
(\mbox{d}\tau^{-1}_{-h\xi_{k-1}})^*D_3l_d(x_{k-1},x_{k},h\xi_{k-1})-(\mbox{d}\tau^{-1}_{h\xi_{k}})^*D_3l_d(x_k,x_{k+1},h\xi_{k})=\\ \lambda_{a,k}\,\eta_t^a(x_k)\,E^t.
\end{split}
\end{equation}
Generally speaking, in most applications one could bypass the definition of $L_d$ and choose $l_d$ to be defined by
\[
l_d(x_k,x_{k+1},\sigma_k)=hL\left(\frac{x_k+x_{k+1}}{2},\frac{x_{k+1}-x_k}{h},\frac{\sigma_k}{h}\right)
\]
or a similar formula.

As we know, \eqref{GNIredb} provides a discretization of the nonholonomic constraints on the Hamiltonian side:
\begin{multline}
A^{i,a}(x_k)\lp\tilde{\met}_{ij}\frac{(x_{k+1}^j-x_{k-1}^j)}{2h}+\frac{1}{2h}\tilde{\met}_{it}((\tau^{-1}(g_{k}))^t
+(\tau^{-1}(g_{k-1}))^t)\rp\\
+\frac{1}{2}B^{t,a}(x_k)\lp\ell^*_{g_{k-1}}D_3L_d(x_{k-1},x_k,g_{k-1})+r_{g_{k}}^*D_3L_d(x_{k},x_{k+1},g_{k})\rp_t=0,\label{Redc}
\end{multline}
or, equivalently,
\begin{multline*}
A^{i,a}(x_k)\lp\tilde{\met}_{ij}\frac{(x_{k+1}^j-x_{k-1}^j)}{2h}+\frac{1}{2}\tilde{\met}_{it}(\xi_{k}^t+\xi_{k-1}^t)\rp\\
+\frac{1}{2}B^{t,a}(x_k)\lp(\mbox{d}\tau^{-1}_{-h\xi_{k-1}})^*D_3l_d(x_{k-1},x_{k},h\xi_{k-1})+(\mbox{d}\tau^{-1}_{h\xi_{k}})^*D_3l_d(x_k,x_{k+1},h\xi_{k})\rp_t=0,
\end{multline*}
where
\begin{align*}
A^{i,a}(x_k)&=(\tilde{\met}^{-1})^{ij}\mu_j^a(x_k)+(\tilde{\met}^{-1})^{it}\eta_t^a(x_k),\\
B^{t,a}(x_k)&=(\tilde{\met}^{-1})^{ti}\mu_i^a(x_k)+(\tilde{\met}^{-1})^{ts}\eta_s^a(x_k),
\end{align*}
$(\tilde{\met}^{-1})$ being the inverse matrix of  $(\tilde{\met})=\left(\begin{array}{ll}
\tilde{\met}_{ij}&\tilde{\met}_{sj}\\\tilde{\met}_{it}&\tilde{\met}_{st}\end{array}\right)$.

Our aim in the following is to find an extension of the nonholonomic RATTLE algorithm presented in remark \ref{remark1} for systems defined on $T\R^n\times \al$. For that purpose we define $\tilde p_k,p_{k+1/2}\in T^*_{x_k}\R^n$ and $\tilde M_k,M_{k+1/2}\in \al^*$ by
\begin{align*}
(\tilde p_k)_i&=\tilde{\met}_{ij}\frac{(x_{k+1}^j-x_{k-1}^j)}{2h}+\frac{1}{2}\tilde{\met}_{is}(\xi_{k}^s+\xi_{k-1}^s),\\
(p_{k+1/2})_i&=\tilde{\met}_{ij}\frac{(x_{k+1}^j-x_k^j)}{h}+\tilde{\met}_{is}\xi_k^s,\\
\tilde M_k&=(\mbox{d}\tau_{h\xi_k}^{-1})^*D_3l_d(x_k,x_{k+1},h\xi_k),\\
M_{k+1/2}&=\Ad^*_{\tau(h\xi_k)}\tilde M_k-\frac{1}{2}\lambda_{a,k+1}\,\eta_s^a(x_{k+1})\,E^s,
\end{align*}
where $\tilde\lambda_{a,k}=\lambda_{a,k}/h$. We also recall that $\xi_k=\tau^{-1}(g_k)/h$. After these redefinitions, equations \eqref{Red}, \eqref{Redd_2} and \eqref{Redc} can be translated into the following algorithm
\begin{subequations}\label{RAT}
\begin{align}
p_{k+1/2}&=\tilde p_k-\frac{h}{2}\lp V_x(x_k)+\tilde\lambda_{a,k}\,\mu^a(x_k)\rp,\label{RATa}\\
M_{k+1/2}&=\Ad^*_{\tau(h\xi_k)}\tilde M_k-\frac{1}{2}\lambda_{a,k+1}\,\eta^a(x_{k+1}),\label{RATb}\\
x_{k+1}^i&=x_k^i+h(\tilde{\met}^{-1})^{ij}\lp(p_{k+1/2})_j-\tilde{\met}_{jt}\,\xi^t_{k}\rp,\label{RATc}\\
A^a(x_{k+1})\,\tilde p_{k+1}&+\frac{1}{2}B^a(x_{k+1})\lp\Ad^*_{\tau(h\xi_{k})}\,\tilde M_{k}+\tilde M_{k+1}\rp=0,\label{RATd}\\
\tilde p_{k+1}&=p_{k+1/2}-\frac{h}{2}\lp V_x(x_{k+1})+\tilde\lambda_{a,k+1}\,\mu^a(x_{k+1})\rp,\label{RATe}\\
\tilde M_{k+1}&=M_{k+1/2}-\frac{1}{2}\lambda_{a,k+1}\,\eta^a(x_{k+1}),\label{RATf}
\end{align}
\end{subequations}
with the natural definitions
$\eta^a(x_k)=\eta^a_t(x_k)\,E^t$, 
$\mu^a(x_k)=\mu^a_i\,dx^i$, 
$A^a(x_k)=A^{i,a}(x_k) \frac{\partial}{\partial x^i}$,
$B^a(x_k)=B^{t,a}(x_k) E_t$; moreover, most of the equations are written in matrix form.

Next, we present the following sequence in order to obtain the 1-step values $(x_{k+1},\tilde p_{k+1},\xi_{k+1},\tilde M_{k+1},\tilde\lambda_{a,k+1})$ from the original values $(x_k,\tilde p_k,\xi_k,\tilde M_k,\tilde\lambda_{a,k})$. First, it is clear that $p_{k+1/2}$ is directly obtained from \eqref{RATa}. Once $p_{k+1/2}$ is fixed, the same happens in \eqref{RATc} determining $x_{k+1}$. Moreover, introducing \eqref{RATb} into \eqref{RATf} we obtain the system of equations
\begin{align*}
0&=A^a(x_{k+1})\,\tilde p_{k+1}+\frac{1}{2}B^a(x_{k+1})\lp\Ad^*_{\tau(h\xi_{k})}\,\tilde M_{k}+\tilde M_{k+1}\rp,\\
\tilde p_{k+1}&=p_{k+1/2}-\frac{h}{2}\lp V_x(x_{k+1})+\tilde\lambda_{a,k+1}\,\mu^a(x_{k+1})\rp,\\
\tilde M_{k+1}&=\Ad^*_{\tau(h\xi_k)}\tilde M_k-\lambda_{a,k+1}\,\eta^a(x_{k+1}),
\end{align*}
which implicitly provides $(\tilde p_{k+1},\tilde M_{k+1},\tilde\lambda_{a,k+1})$. Therefore, we see that equations \eqref{RAT} do not give the value $\xi_{k+1}$ directly. Nevertheless, replacing \eqref{RATa} into \eqref{RATc} and taking a step forward, we obtain the equation
\[
x_{k+2}=x_{k+1}+h(\tilde{\met}^{-1})\lp\tilde p_{k+1}-\frac{h}{2}\lp V_x(x_{k+1})+\tilde\lambda_{a,k+1}\,\mu^a(x_{k+1})\rp-\tilde{\met}\,\xi_{k+1}\rp,
\]
which determines $x_{k+2}$ in terms of $x_{k+1},\tilde p_{k+1},\tilde\lambda_{a,k+1}$ (already fixed by the previous sequence) and $\xi_{k+1}$. Finally, introducing this value of $x_{k+2}$ into the definition of $\tilde M_{k+1}$ we obtain the equation
\[
\tilde M_{k+1}=(\mbox{d}\tau_{h\xi_{k+1}}^{-1})^*D_3l_d(x_{k+1},x_{k+2},h\xi_{k+1}),
\]
which implicitly determines $\xi_{k+1}$ since $\tilde M_{k+1}$ has been previously determined. Note that this last step is not incompatible with equations \eqref{RAT} since the chosen value of $x_{k+2}$, and also $\tilde M_{k+1}$'s, is precisely the one that the algoritm provides. Schematically, the proposed algorithm can be represented by
\begin{equation*}
(x_k,\tilde p_k,\xi_k,\tilde M_k,\tilde\lambda_{a,k})\Flder(p_{k+1/2},x_{k+1})\Flder(\tilde p_{k+1},\tilde M_{k+1},\tilde\lambda_{a,k+1})\Flder \xi_{k+1}.
\end{equation*}

\begin{remark}
{\rm
A natural question related to the reduction of continuous or discrete mechanical systems with symmetry concerns the reverse procedure. Once the solutions of the reduced system have been obtained, how can we recover from them the solutions of the unreduced system? Observe that, in our case, we have only considered the case of trivial principal bundles 
$\hbox{pr}_1: M\times G\to M$ with trivial action $\Phi_{\tilde{W}}(x, W)=(x, \tilde{W}W)$ where $x\in M$ and $W, \tilde{W}\in G$.
The original mechanical Lagrangian is defined by $L: T(M\times G)\equiv TM\times TG\rightarrow {\mathbb R}$ along with the  nonholonomic distribution ${\mathcal D}$.
The reduced system $(\tilde{L}, \tilde{\mathcal D})$ is defined on $TM\times {\mathfrak g}$ and, given a reduced solution of the nonholonomic system  $(x(t), \xi(t))$, we cane obtain the solution of the original system by solving additionally the equation $\dot{W}(t)=W(t)\xi(t)$, which is called the {\it reconstruction equation}. 
In the discrete case we have a similar scheme. Namely, a reduced solution is a sequence $(x_k, x_{k+1}, g_k)$ and the discrete solutions  $(x_k, x_{k+1}, W_k, W_{k+1})$ of the unreduced system 
are derived by the discrete reconstruction equation $W_{k+1}=W_k g_k$. Moreover, if we describe our reduced integrator using a retraction map $\tau: {\mathfrak g}\rightarrow G$, then 
the reconstruction equation reads $W_{k+1}=W_k\tau(h\xi_k)$. 
}
\end{remark}

\subsection{The case of affine constraints}\label{aff-red}

We consider in this section the extension of the reduced GNI method for the case of affine nonholonomic constraints. With the same notation as in section \ref{wert}, take an affine bundle $\tilde{\mathcal A}$ of $TM\times \al $ modeled on the vector bundle $\tilde{\mathcal D}$ and  assume that there exists a globally defined section $\tilde{Y}\in \Gamma(TM\times \al)$ such that $v_x\in \tilde{{\mathcal A}}_x$ if and only if
 $v_x-\tilde{Y}(x)\in \tilde{{\mathcal D}}_x$.

 Fixing a local basis of sections $\{e_I\}=\{\tilde{e}_\alpha, \tilde{e}_a\}$ of $\Gamma(TM\times \al)$ adapted to the orthogonal decomposition $\tilde{{\mathcal D}}\oplus \tilde{{\mathcal D}}^\perp$, the constraints determining locally the affine subbundle $\tilde{\mathcal A}$ are
 \[
 y^a-Y^a(x)=0
 \]
 where $\tilde{Y}=Y^{\alpha} \tilde{e}_{\alpha}+Y^a\tilde{e}_a$.

In our case, the initial data defining our {\it reduced nonholonomic affine problem} is denoted by the 4-tuple $(\tilde{\mathcal{D}}, \tilde{\mathcal G}, \tilde{Y}, \tilde{V})$ (see section \ref{wert}). By means of the metric, from $\tilde{Y}$, we can uniquely define a 1-section  $\tilde{\mathcal G}(\tilde Y,\cdot)=\Pi\in\Gamma(T^*M\times \al^*)$. Locally, $\Pi=\tilde{\mathcal G}_{IJ}Y^Je^I$.

Consider a discrete Lagrangian $L_d:Q\times Q\times G\Flder\R$. 
As in the previous sections,  we can split the total space $E$ as $E=\tilde\dist\oplus\tilde\dist^{\perp}$ with  corresponding projectors $\pe:E\Flder\tilde\dist$, $\qu:E\Flder\tilde\dist^{\perp}$. 
Thus, the proposed {\bf reduced GNI equations} for affine constraints are a mixture of definitions \ref{GNIAff} and \ref{GNIRed}, namely
\begin{subequations}\label{GNIred-1}
\begin{align}
&\mathcal{P}_{x_k}^*\lp \F L_d^-(x_k,x_{k+1},g_k)\rp-\mathcal{P}_{x_{k}}^*\lp \F L_d^+(x_{k-1},x_{k},g_{k-1})\rp=0,\label{GNIreda-aff}\\
&\mathcal{Q}_{x_k}^*\lp \F L_d^-(x_k,x_{k+1},g_k)\rp+\mathcal{Q}_{x_{k}}^*\lp \F L_d^+(x_{k-1},x_{k},g_{k-1})\rp+2\mathcal{Q}_{x_{k}}^*\Pi=0,\label{GNIredb-aff}
\end{align}
\end{subequations}
where the Legendre transforms $\F L_d^{\pm}$ are defined in \eqref{RedLegTr}.

\subsection{Example: the rolling ball}

Consider the motion of an inhomogeneous sphere of radius $r>0$ that rolls without slipping
on a table. We will consider two cases: when the table is fixed or when it is rotating with constant angular velocity $\Omega$ around a vertical axis. The first one corresponds to linear nonholonomic constraints  while the second to affine ones.
 
If the center of mass of the sphere coincides with the geometric center, we recover
the well-known problem of the Chaplygin sphere, which possesses an invariant measure. The general case is known  as the Chaplygin top and its qualitative behaviour is quite different depending on the cases exposed. For instance, it is known that  the Chaplygin top has an invariant measure if and only if:
(i) The center of mass of the sphere coincides with the geometric center or (ii) The ball is axially symmetric (see  \cite{BM1,BMB,BMK,FGM} and references therein).


The configuration space for the continuous system is $Q=\R^2\times\gru$ and we shall use the notation $(x,y;R)$ to represent a typical point in $Q$. In consequence, according to the previous subsection, $E=T\R^2\times\alg$. It is well-known that there exists an isomorphism $\hat\cdot:\R^3\Flder\alg$ given by
\begin{equation}\label{hatso}
\hat\omega=\lp\begin{array}{ccc}
0&-\omega_{3}&\omega_{2}\\
\omega_{3}&0&-\omega_{1}\\
-\omega_{2}&\omega_{1}&0
\end{array}
\rp\in\alg,
\end{equation}
where, obviously, $\omega=(\omega_1,\omega_2,\omega_3)^T\in\R^3$. Given $\dot x\,\frac{\der}{\der x}+\dot y\,\frac{\der}{\der y}\in T\R^2$ and $\hat\omega\in\alg$, the nonholonomic constraints read
\begin{subequations}\label{Con}
\begin{align}
\dot x-r\,\omega_2+\Omega\,y&=0,\label{Cona}\\
\dot y+r\,\omega_1-\Omega\,x&=0.\label{Conb}
\end{align}
\end{subequations}
It is clear that the constraints above do not correspond to the linear case but to the affine one. We will apply the procedure developed in section \ref{aff-red}. Hence, the nonholonomic setting in this example corresponds to a mixture of the settings presented in sections \ref{AffS} and \ref{RS}. Let us define a global basis of sections of $T\R^2\times\alg\Flder\R^2$:
\begin{align*}
&\tilde e_1=\left(\frac{\der}{\der x},0\right)\,,\,\tilde e_2=\left(\frac{\der}{\der y},0\right),\\
&\tilde e_3=(0,E_1)\,,\,\tilde e_4=(0,E_2)\,,\,\tilde e_5=(0,E_3),
\end{align*}
where $\lc E_1,E_2,E_3\rc$ is the basis of $\alg$ obtained from the standard basis of $\R^3$ via the isomorphism $\hat\cdot$. Therefore, the distribution generated by the constraints \eqref{Con} may be written in this basis as
\[
\tilde{\dist}=\spanop\lc r\tilde e_1+\tilde e_4\,,\,-r\tilde e_2+\tilde e_3\,,\,\tilde e_5\rc,
\]
while the vector field $\tilde Y$ is
\[
\tilde{Y}=-\Omega\,y\,\frac{\der}{\der x}+\Omega\,x\,\frac{\der}{\der y}\in\mathfrak{X}(\R^2).
\]
Moreover, the reduced Lagrangian function $l:T\R^2\times\alg\Flder\R$ is given by the kinetic energy, i.e.,
\begin{equation}\label{LagCont}
l(x,y,\dot x,\dot y;\omega)=\frac{1}{2}m(\dot x^2+\dot y^2)+\frac{1}{2}\,(I_1\,\omega_1^2+I_2\,\omega_2^2+I_3\,\omega_3^2),
\end{equation}
where $\dot q=(x,y;\dot x,\dot y)^T\in T\R^2$ and $\omega\in\alg$ (where, as mentioned before, we are employing the isomorphism $\hat\cdot:\R^3\Flder\alg$), which determines the metric
\[
\tilde{\met}=m_{ij}(\tilde e^i\otimes\tilde e^j)+I_{ts}(\tilde e^t\otimes\tilde e^s),
\]
where $i,j=1,2$; $t,s=3,4,5$; $(m_{ij})=\mbox{diag}(m,m)$ and $(I_{ts})=\mbox{diag}(I_1,I_2,I_3)$. With respect to this metric, the orthogonal complement to $\tilde\dist$ is
\[
\tilde\dist^{\perp}=\spanop\lc\tilde e_1-\frac{mr}{I_2}\tilde e_4\,,\,\tilde e_2+ \frac{mr}{I_1}\tilde e_3\rc.
\]
As mentioned above, the example of the rolling ball fits in an affine nonholonomic scheme for the reduced system $T\R^2\times\alg$. Define the discrete reduced Legendre transformations
\[
\F l_d^{\pm}:\R^2\times\R^2\times\alg\Flder T^*\R^2\times\dalg,
\]
as
\begin{small}
\begin{eqnarray*}
\F l_d^-\colon \R^2\times \R^2\times\alg&\to & T^*\R^2\times\dalg\\
(q_k, q_{k+1},\omega_k)&\longmapsto& (q_k, -D_1 l_d(q_k, q_{k+1},\omega_k),(\mbox{d}\tau^{-1}_{h\omega_k})^*D_3l_d(q_k, q_{k+1},\omega_k)),\\
\F l_d^+\colon  \R^2\times \R^2\times \alg&\to&  T^*\R^2\times\dalg\\
(q_k, q_{k+1},\omega_k)&\longmapsto& (q_{k+1}, D_2 l_d(q_k, q_{k+1},\omega_k),(\mbox{d}\tau^{-1}_{-h\omega_k})^*D_3l_d(q_k, q_{k+1},\omega_k))\;,
\end{eqnarray*}
\end{small}%
where the relationship between $\F L_d^{\pm}$ and $\F l_d^{\pm}$ is given by the properties of the retraction map $\tau$ presented in Appendix (see \cite{JKM} for more details).
The proposed nonholonomic equations \eqref{GNIred} become
\begin{subequations}\label{RB}
\begin{align}
\mathcal{P}_{q_k}^*\lp \F l_d^-(q_k,q_{k+1},\omega_k)\rp&-\mathcal{P}_{q_{k}}^*\lp \F l_d^+(q_{k-1},q_{k},\omega_{k-1})\rp=0,\label{RBa}\\
\mathcal{Q}_{q_k}^*\lp \F l_d^-(q_k,q_{k+1},\omega_k)\rp&+\mathcal{Q}_{q_{k}}^*\lp \F l_d^+(q_{k-1},q_{k},\omega_{k-1})\rp+2\mathcal{Q}_{q_k}^*\Pi=0,\label{RBb}
\end{align}
\end{subequations}
where $q_k=(x_k,y_k)^T$, $\omega_k\in\alg$ and $\Pi=\tilde\met(\tilde Y,\cdot)$, which in this case reads
\[
\Pi=-m\Omega y\,\tilde e^1+m\Omega x\,\tilde e^2.
\] 
We choose the discrete Lagrangian $l_d:\R^2\times\R^2\times\alg\Flder\R$ as $l_d(q_k,q_{k+1},\omega_k)=hl(q_k,\frac{q_{k+1}-q_k}{h},\omega_k)$, namely
\begin{small}
\begin{equation}\label{LagDisc}
l_d(q_k,q_{k+1},\omega_k)=\frac{m}{2h}\lp(x_{k+1}-x_k)^2+(y_{k+1}-y_k)^2\rp+\frac{h}{2}\lp I_1(\omega_k)_1^2+I_2(\omega_k)_2^2+I_3(\omega_k)_3^2\rp.
\end{equation}
\end{small}
Finally, the projection $\qu:T\R^2\times\alg\Flder\tilde\dist^{\perp}$ is given in coordinates by the matrix
\begin{equation}\label{proQ}
\qu=\lp\begin{array}{ccccc}
\frac{I_2}{mr^2+I_2}&0&0&\frac{-r\,I_2}{mr^2+I_2}&0\\
0&\frac{I_1}{mr^2+I_1}&\frac{r\,I_1}{mr^2+I_1}&0&0\\
0&\frac{mr}{mr^2+I_1}&\frac{mr^2}{mr^2+I_1}&0&0\\
\frac{-mr}{mr^2+I_2}&0&0&\frac{mr^2}{mr^2+I_2}&0\\
0&0&0&0&0
\end{array}\rp,
\end{equation}
while $\mathcal{P}:T\R^2\times\alg\Flder\tilde\dist$ is given by
\begin{equation}\label{proP}
\mathcal{P}=\lp\begin{array}{ccccc}
\frac{mr^2}{mr^2+I_2}&0&0&\frac{r\,I_2}{mr^2+I_2}&0\\
0&\frac{mr^2}{mr^2+I_1}&\frac{-r\,I_1}{mr^2+I_1}&0&0\\
0&\frac{-mr}{mr^2+I_1}&\frac{I_1}{mr^2+I_1}&0&0\\
\frac{mr}{mr^2+I_2}&0&0&\frac{I_2}{mr^2+I_2}&0\\
0&0&0&0&1
\end{array}\rp,
\end{equation}
Setting the retraction map $\tau$ as the Cayley map for $SO(3)$, that is $\tau(\omega)=\mbox{cay}(\omega)$ (see Appendix for more details) and taking into account \eqref{LagDisc}, \eqref{proQ}, \eqref{proP}; then equations \eqref{RB} read
\begin{align*}
mr\lp\frac{x_{k+1}-2x_k+x_{k-1}}{h}\rp+I_2\lp\omega^k_2-\omega^{k-1}_2\rp+\frac{h}{2}(I_1-I_3)\lp\omega^k_1\omega^k_3+\omega^{k-1}_1\omega^{k-1}_3\rp+O_1(h^2)&=0,\\\\
mr\lp\frac{y_{k+1}-2y_k+y_{k-1}}{h}\rp-I_1\lp\omega^k_1-\omega^{k-1}_1\rp-\frac{h}{2}(I_3-I_2)\lp\omega^k_2\omega^k_3+\omega^{k-1}_2\omega^{k-1}_3\rp+O_2(h^2)&=0,\\\\
I_3\lp\omega^k_3-\omega^{k-1}_3\rp+\frac{h}{2}(I_2-I_1)\lp\omega^k_1\omega^k_2+\omega^{k-1}_1\omega^{k-1}_2\rp+O_3(h^2)&=0,\\\\
\frac{x_{k+1}-x_{k-1}}{2h}+\Omega\,y_k-r\frac{\omega^k_2+\omega^{k-1}_2}{2}-r\frac{h}{4}\frac{I_1-I_3}{I_2}\lp\omega^k_1\omega^k_3-\omega^{k-1}_1\omega^{k-1}_3\rp+O_4(h^2)&=0,\\\\
\frac{y_{k+1}-y_{k-1}}{2h}-\Omega\,x_k+r\frac{\omega^k_1+\omega^{k-1}_1}{2}+r\frac{h}{4}\frac{I_3-I_2}{I_1}\lp\omega^k_2\omega^k_3-\omega^{k-1}_2\omega^{k-1}_3\rp+O_5(h^2)&=0,
\end{align*}
where
\begin{eqnarray*}
&&O_1(h^2)=\frac{h^2}{4}\lc \omega_2^k\,||\omega^k||_I^2-\omega_2^{k-1}\,||\omega^{k-1}||^2_I\rc,\\\\
&&O_2(h^2)=-\frac{h^2}{4}\lc \omega_1^k\,||\omega^k||^2_I-\omega_1^{k-1}\,||\omega^{k-1}||^2_I\rc,\\\\
&&O_3(h^2)=\frac{h^2}{4}\lc \omega_3^k\,||\omega^k||^2_I-\omega_3^{k-1}\,||\omega^{k-1}||^2_I\rc,\\\\
&&O_4(h^2)=-\frac{h^2}{8\,I_2}r\lc \omega_2^k\,||\omega^k||^2_I+\omega_2^{k-1}\,||\omega^{k-1}||^2_I\rc,\\\\
&&O_5(h^2)=\frac{h^2}{8\,I_1}r\lc \omega_1^k\,||\omega^k||^2_I+\omega_1^{k-1}\,||\omega^{k-1}||^2_I\rc,
\end{eqnarray*}
where $||\omega^k||_I^2=I_1\,(\omega_1^k)^2+I_2\,(\omega_2^k)^2+I_3\,(\omega_3^k)^2$ (equivalently in the case $k-1$). In these equations we recognize an order-one consistent discrete scheme for the continuous equations of the rolling ball system. This fact is not surprising since the discrete Lagrangian \eqref{LagDisc} is as well an order-one approximation of the action integral defined by the continuous Lagrangian \eqref{LagCont} (see \cite{MarsdenWest,PatCuel} for more details regarding the relationship between the order of consistency of the discrete Lagrangian with respect to the action integral and of the variational integrators obtained from them).

In Figure \ref{simulation-homogeneous} we show the numerical results of applying this discrete method. As a first example, we consider a homogeneous ball with $I_1=I_2=I_3=2/3$, and $m=r=\Omega=1$. We take decreasing values of the time step $h$, and compare to the method in \cite{IMMM}. We show errors with respect to the exact solution to the continuous system, with initial conditions $(x_0,y_0, \dot x_0, \dot y_0)=(1,1,1,1)$, $\omega=(0,2,0)$, and a total run time of 10. Figure \ref{trajectory-homogeneous} shows the evolution of the $x_k$, $y_k$ variables for these same physical parameters and initial conditions, for a total run time of 1000.

\begin{figure}[ht]
\begin{center}
\includegraphics[trim = 25mm 0mm 25mm 0mm, clip, scale=.45]{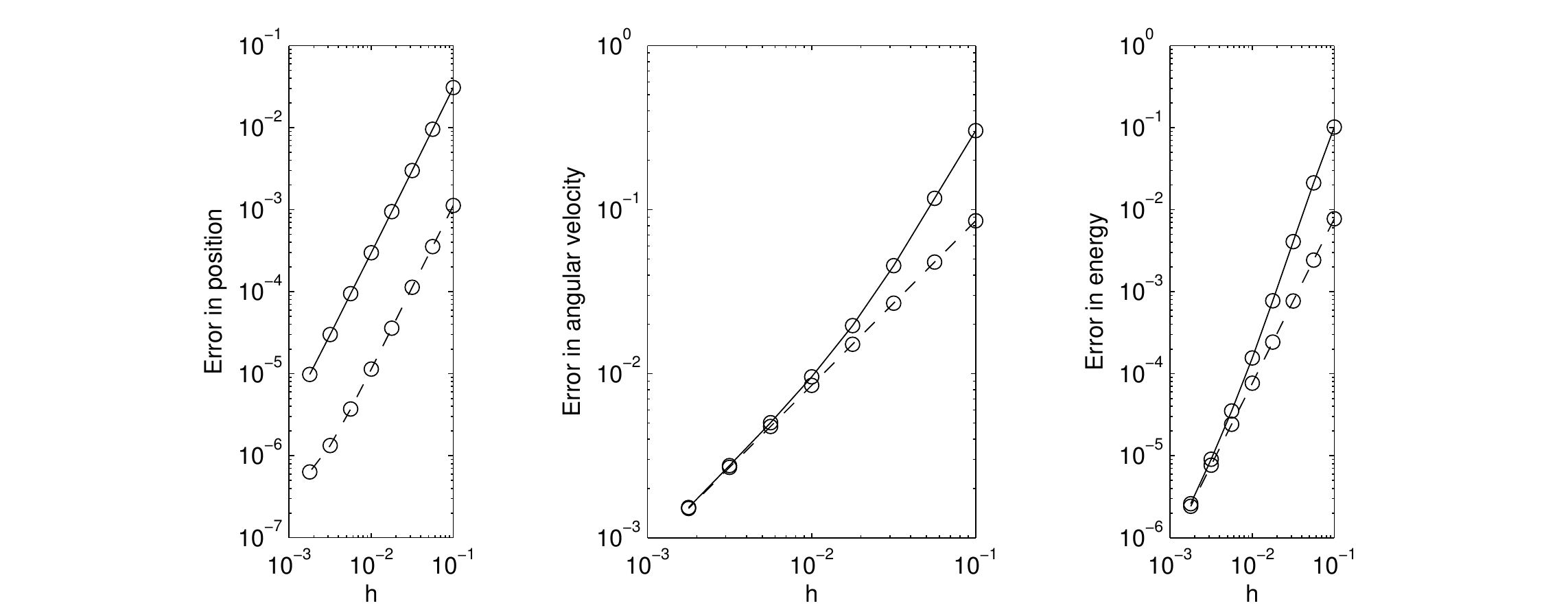}
\end{center}
\caption{Errors in position $(x,y)$, angular velocity $\omega$ and energy. The continuous line corresponds to the proposed method. The dashed line corresponds to the method in \cite{IMMM}.}
\label{simulation-homogeneous}
\end{figure}

\begin{figure}[ht]
\begin{center}
\includegraphics[trim = 25mm 5mm 25mm 0mm, clip, scale=.45]{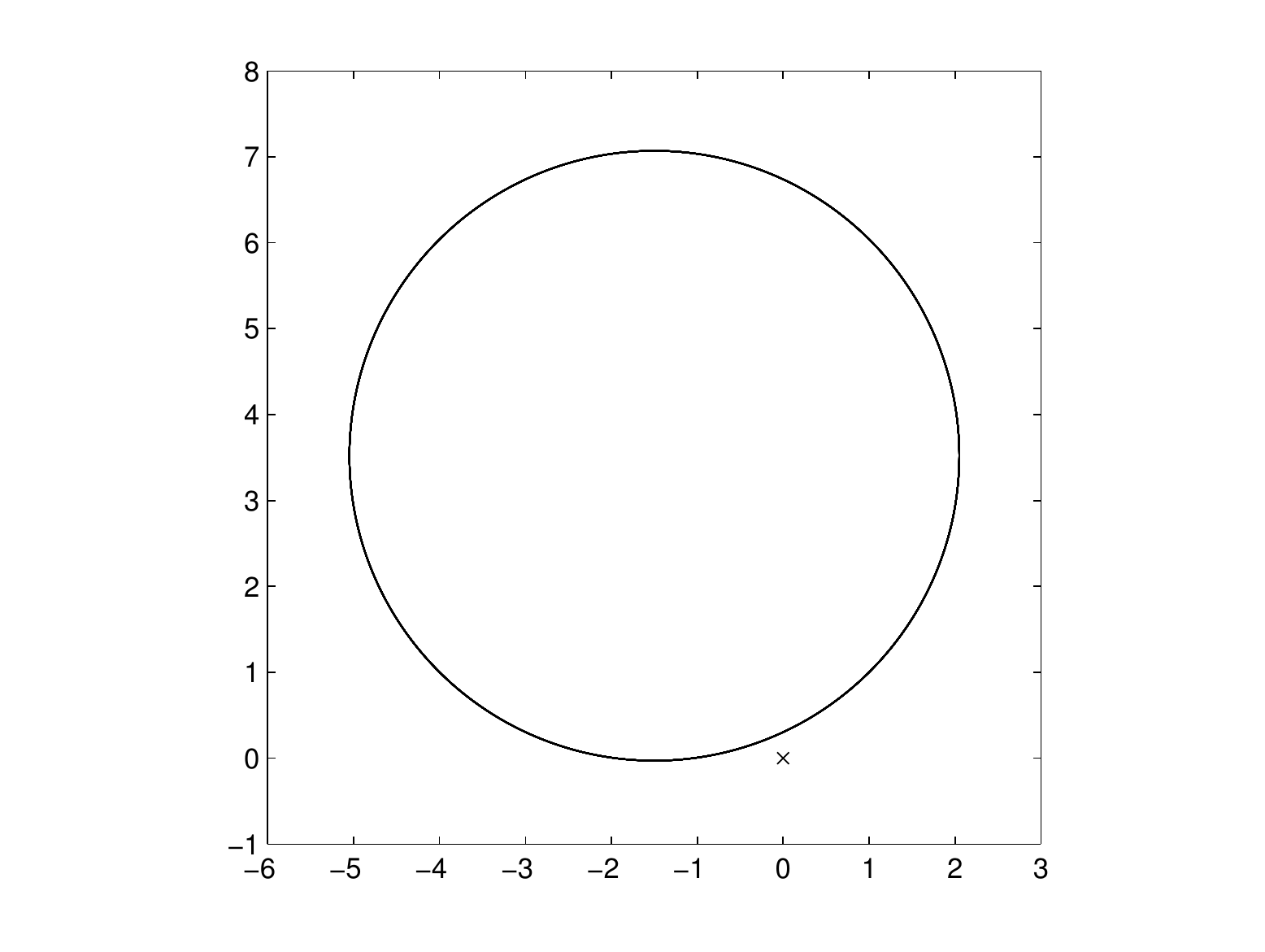}
\end{center}
\caption{Position $(x,y)$ for the point of contact of the homogeneous ball, with $h=0.1$ and 10000 steps. The cross indicates the center of the rotating plate.}
\label{trajectory-homogeneous}
\end{figure}

Our second, more general example uses $I_1=1$, $I_2=1.1$, $I_3=1.2$, $m=3$, $r=1$, $\Omega=0.2$. The initial conditions for the discrete algorithm are $(x_0,y_0)=(1,0)$, $\omega^0=(-0.2, 0, 0.4)$, with $(x_1,y_1)$ obtained from a two-point version of the last two equations of the discrete method above, which represent the constraint equations. The simulation was run up to a final time of 15, with $h$ decreasing exponentially from 0.15 to around $8.4 \times 10^{-5}$. A discrete trajectory with a smaller value of $h\approx 2.7\times 10^{-5}$ was used as a reference solution against which we computed errors in position and velocity in order to visualize the behavior of the method regarding convergence (Figure \ref{simulation-convergence-example2}).

\begin{figure}[ht]
\begin{center}
\includegraphics[trim = 33mm 0mm 43mm 0mm, clip, scale=.44]{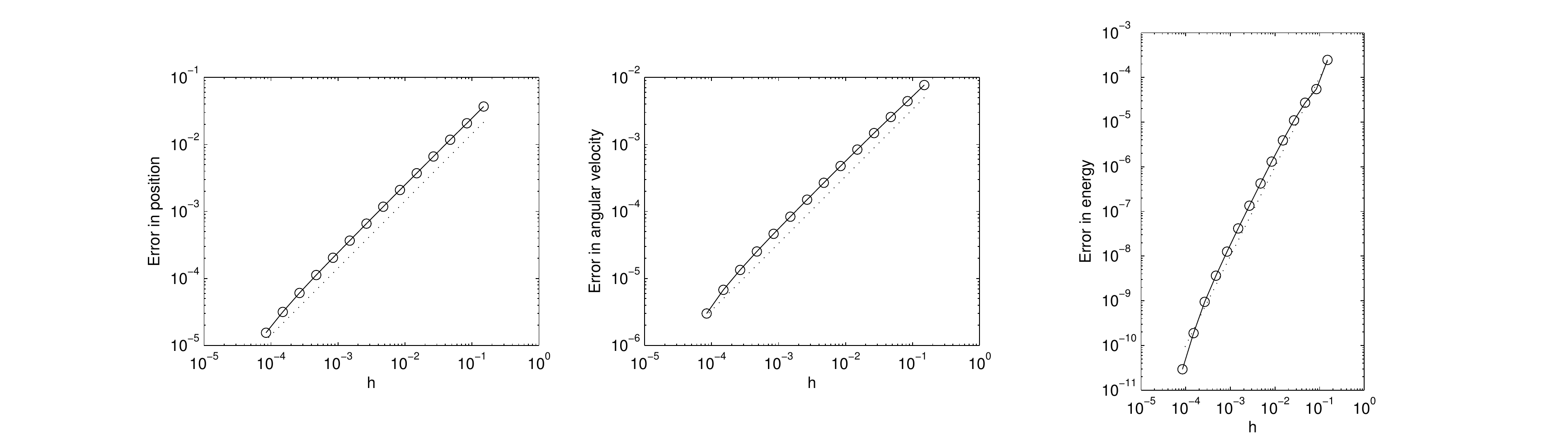}
\end{center}
\caption{Errors in position, angular velocity and energy. Dotted lines are multiples of $h$, $h$ and $h^2$ respectively, for comparison of slopes.}
\label{simulation-convergence-example2}
\end{figure}

\section{Extension to Lie algebroids}\label{LieAlg}

\subsection{Brief introduction to Lie groupoids and Lie algebroids}

\begin{definition}
A \textbf{Lie groupoid}, denoted $G\rightrightarrows Q$, consists of two differentiable manifolds  $G$ and $Q$ and the
  following differentiable maps (the structural maps):
\begin{enumerate}
\item A pair of submersions: the  \textbf{source map} $ \alpha \colon G \rightarrow Q $ and the \textbf{target map} $
  \beta \colon G \rightarrow Q $.
\item An associative \textbf{multiplication map} $ m \colon G _2 \rightarrow G
  $, $ (g,h) \mapsto gh $, where the set
\[
  G _2 = \left\{ \left( g , h
    \right) \in G \times G \;\middle\vert\; \beta (g) = \alpha (h)
  \right\}
\]
is called the set of \textbf{composable pairs}.
\item An \textbf{identity section} $ \epsilon \colon Q \rightarrow G $ of $\alpha$ and $\beta$, such
  that for all $ g \in G $,
\[
  \epsilon \left( \alpha (g) \right) g = g = g \epsilon \left( \beta
    (g) \right).
\]
\item An \textbf{inversion map} $ i \colon G \rightarrow G $, $ g \mapsto
  g ^{-1} $, such that for all $ g \in G $,
\[
  g g ^{-1} = \epsilon \left( \alpha (g) \right) , \qquad g ^{-1} g =
  \epsilon \left( \beta (g) \right) .
\]
\end{enumerate}
\end{definition}

Next, we will introduce the notion of a left (right) translation
by an element of a Lie groupoid. Given a groupoid $ G \rightrightarrows Q $ and an element $ g \in G
  $, define the \textbf{left translation} $ \ell _g \colon \alpha ^{-1}
  \left( \beta (g) \right) \rightarrow \alpha ^{-1} \left( \alpha (g)
  \right) $ and \textbf{right translation} $ r _g \colon \beta ^{-1}
  \left( \alpha (g) \right) \rightarrow \beta ^{-1} \left( \beta (g)
  \right) $ by $g$ to be
\[
  \ell _g (h) = g h , \qquad r _g (h) = h g .
\]

Analogously to the case of Lie groups, one may introduce the notion of left (right)-invariant vector field in a Lie groupoid from these translations.
Given a Lie groupoid $ G \rightrightarrows Q $, a vector field $ \xi
  \in \mathfrak{X} (G) $ is \textbf{left-invariant} if $\xi$ is $\alpha$-vertical
  and $ \left( T _h \ell_g \right) \left( \xi (h) \right) = \xi \left( g h \right) $ for
  all $(g, h) \in G _2 $.  Similarly, $\xi$ is \textbf{right-invariant}
  if $\xi$ is $\beta$-vertical and
  $ \left( T _h r _g \right) \left(  \xi (h) \right) = \xi \left( h g
  \right) $ for all $ (h, g) \in G _2 $.

It is well known that there always exists a Lie algebroid associated to a Lie groupoid (again analogously to the Lie group case).
We consider the vector bundle $\tau_{_{AG}}: AG \to Q$, whose fiber
at a point $x \in Q$ is $(AG)_{x} = V_{\epsilon(x)}\alpha
=\mbox{ker}\,(T_{\epsilon(x)}\alpha)$. It is easy to prove that there
exists a bijection between the space $\Gamma(\tau)$ and the set of
left (right)-invariant vector fields on $G$.
If $X$ is a section of $\tau_{_{AG}}: AG \to Q$, the corresponding
left (right)-invariant vector field on $G$
will be denoted $\lvec{X}$ (respectively, $\rvec{X}$), where
\begin{align*}
\lvec{X}(g) &= (T_{\epsilon(\beta(g))}\ell_{g})(X(\beta(g))),\\
\rvec{X}(g) &= -(T_{\epsilon(\alpha(g))}r_{g})((T_{\epsilon
(\alpha(g))}i)( X(\alpha(g)))),
\end{align*}
for $g \in G$. Using the above facts, we may introduce a {\bf Lie
algebroid structure} $(\lcf\cdot , \cdot\rcf, \rho)$ on
$AG$, which is defined by
\[
\lvec{\lcf X, Y\rcf} = [\lvec{X}, \lvec{Y}], \makebox[.3cm]{}
\rho(X)(x) = (T_{\epsilon(x)}\beta)(X(x)),
\]
for $X, Y \in \Gamma(\tau)$ and $x \in Q$. Note that
\[
\rvec{\lcf X, Y\rcf} = -[\rvec{X}, \rvec{Y}], \makebox[.3cm]{}
[\rvec{X}, \lvec{Y}] = 0,
\]
(for more details, see \cite{Mac}).

\subsection{GNI extension to Lie groupoids}Let $G\rightrightarrows Q$ be a Lie groupoid and $\tau_{_{AG}}: AG \to Q$ its associated Lie algebroid.
Consider a mechanical system subjected to linear nonholonomic constraints, that is, a pair $(L,\mathcal{D})$ (see \cite{MMM,IMMM} for more details), where
\begin{itemize}
\item[$i)$] $L:AG\Flder\R$ is a Lagrangian function of mechanical type
\[
L(a)=\frac{1}{2}\,\mathcal{G}(a,a)-V(\tau_{_{AG}}(a)),\,\,\,\mbox{where}\,a\in AG.
\]
\item[$ii)$] $\mathcal{D}$ is the total space of a vector subbundle $\tau_{\mathcal{D}}:\mathcal{D}\Flder Q$ of $AG$.
\end{itemize}
Here $\mathcal{G}:AG\times_{Q}AG\Flder\R$ is a bundle metric on $AG$. We also consider the orthogonal decomposition $AG=\mathcal{D}\oplus\mathcal{D}^{\perp}$ and the associated projectors 
\begin{equation}\label{AlgPro}
\mathcal{P}:AG\Flder\mathcal{D}\quad\quad\quad \mbox{and}\quad\quad\quad \mathcal{Q}:AG\Flder\mathcal{D}^{\perp}.
\end{equation}
Consider a discretization $L_d:G\Flder\R$ of the Lagrangian $L$. It is possible to define two Legendre transformations $\F L_d^{\pm}\colon G\to A^*G$ by
\begin{equation}\label{AlgLegTr}
\begin{split}
  \F L_d^-(h)(v_{\epsilon(\alpha(h))})&=-v_{\epsilon(\alpha(h))}(L_d\circ r_h \circ i),\\
  \F L_d^+(g)(v_{\epsilon(\beta(g))})&=v_{\epsilon(\beta(g))}(L_d\circ \ell_g),
\end{split}
\end{equation}
where $v_{\epsilon(\alpha(h))}\in A_{\alpha(h)}G$ and $v_{\epsilon(\beta(g))}\in A_{\beta(g)}G$. Therefore $\F L_d^-(h)\in A^*_{\alpha(h)}G$ and $\F L_d^+(g)\in A^*_{\beta(g)}G$. Since the Euler-Lagrange equations are given by the matching of momenta, in the Lie groupoid setting they read
\[ \F L_d^-(h)=\F L_d^+(g), \]
where $(g,h)$ is in the set $G_2$.

\begin{definition}\label{GNIAlg}
Consider the projectors \eqref{AlgPro} and the discrete Legendre transforms $\F L_d^{\pm}$ \eqref{AlgLegTr}. The extension of the GNI method for Lie algebroids is defined by the equations
\begin{subequations}\label{GNIgrup}
\begin{align}
\mathcal{P}^*_q\lp \F L_d^-(h)-\F L_d^+(g) \rp&=0\\
\mathcal{Q}^*_q\lp \F L_d^-(h)+\F L_d^+(g) \rp&=0,
\end{align}
\end{subequations}
where the subscript $q$ emphasizes the fact that the projections take place in the fiber over $q=\alpha(h)=\beta(g)$. 
\end{definition}
Let $\lc X_{\alpha},X_a\rc$ be a local basis adapted to $\mathcal{D}\oplus\mathcal{D}^{\perp}$, in the sense that locally $\mathcal{D}=\mbox{span}\lc X_{\alpha}\rc$ and $\mathcal{D}^{\perp}=\mbox{span}\lc X_{a}\rc$. We can rewrite equations \eqref{GNIgrup} as
\begin{subequations}\label{GNIgrupLoc}
\begin{align}
 \F L_d^-(h)\lp X_{\alpha}(q)\rp-\F L_d^+(g)\lp X_{\alpha}(q)\rp &=0,\label{GNIgrupLoca}\\
 \F L_d^-(h)\lp X_{a}(q)\rp+\F L_d^+(g)\lp X_{a}(q)\rp&=0\label{GNIgrupLocb},
\end{align}
\end{subequations}
where $\alpha(h)=\beta(g)=q\in Q$ (so $(g,h)\in G_2$). Let us denote
\begin{align*}
p_g^+&=\F L_d^+(g)\in A^*_qG,\\
p_h^-&=\F L_d^-(h)\in A^*_qG,
\end{align*}
so equation \eqref{GNIgrupLocb} becomes
\[
\lp\frac{p_g^++p_h^-}{2}\rp\lp X_a(q)\rp=0.
\]
If $\mu^a\in\Gamma(A^*G)$ are such that $\mathcal{D}^{\circ}=\mbox{span}\lc\mu^a\rc$, then this last equation becomes
\[
\mathcal{G}\lp\frac{p_g^++p_h^-}{2},\mu^a\rp=0,
\]
where, by a slight abuse of notation, we denote the bundle metric on $A^*G$ naturally induced by the bundle metric on $AG$ using the same symbol $\mathcal{G}$. Note that the set of $\eta\in A^*G$ such that $\mathcal{G}(\eta,\mu^a)=0$ for all $a$ forms the constraint submanifold $\mathcal{\bar D}=\mbox{Leg}_{\mathcal{G}}(\mathcal{D})$. Therefore the average momentum $\tilde p=(p_g^++p_h^-)/2\in\mathcal{\tilde D}$ satisfies in this sense the constraint equations.

\section{Conclusions}
In this paper,  we continue the study of the properties of the  Geometric Nonholonomic Integrator (GNI) and extending the construction given in our previous work \cite{SDD1} to a more extense class of nonholonomic systems (reduced systems and systems with affine constraints). Our paper shows the importance of combining different research areas
(differential geometry, numerical analysis and mechanics) to produce
methods with an extraordinary qualitative and quantitative behavior.

Such issues raise a number of future work directions. We therefore
close with some open questions and future work:
\begin{itemize}
\item Given a Geometric Nonholonomic Integrator,  does
  there exist, in the sense of backward error analysis, a continuous
  nonholonomic system, such that the discrete evolution for the
  nonholonomic integrator is the  flow of this  nonholonomic system up
  to an appropriate order?
\item  Is it possible to use the GNI in order to design numerical methods for optimal control of nonholonomic systems using the techniques developed in \cite{JKM}? Furthermore, with these methods is even  possible to approximate piecewise-smooth control, giving a more realistic behavior. See also \cite{Bl2003, BKM,SMY}.
\item Construction of new methods that mimic the  so-called ``sister'' piecewise holonomic system and study its relationship with the GNI method. 
The study of ``sister'' systems is interesting to modelize the dynamics of human walking, and in an averaged sense they approach to nonholonomic systems (see for more information \cite{HF,Ruina,SrRu,PrHo} and references therein). Observe that GNI is related to an elastic impact with the nonholonomic distribution (see \cite{SDD1}).
\end{itemize}

\section*{Appendix: Retraction maps}

As mentioned in subsection \ref{RS1} a {\it retraction map} $\tau:\al\Flder G$ is an analytic local diffeomorphism which maps a neighborhood of $0\in\al$ onto a neighborhood of the neutral element $e\in G$, such that $\tau(0)=e$ and $\tau(\xi)\tau(-\xi)=e$, for $\xi\in\al$. There are many choices for the map $\tau$ such as the Cayley map, the exponential map, etc.
The retraction map  is used
to express small discrete changes in the group configuration
through unique Lie algebra elements, say $\xi_k=\tau^{-1}(g_k^{-1}g_{k+1})/h$. That is, if $\xi_{k}$ were regarded as an average velocity between $g_{k}$ and $g_{k+1}$, then $\tau$ is an approximation to the integral flow of the dynamics. The difference $g_{k}^{-1}\,g_{k+1}\in G$, which is an element of a nonlinear space, can now be represented by the vector $\xi_{k}$. (See \cite{BouMarsden,MarinThesis} for further details.)

Of great importance is the {\it right trivialized tangent} of the retraction map. Complementary to \eqref{RTT} is the following definition:

\begin{definition}\label{Retr_again}
Given a retraction map $\tau\colon\mathfrak{g}\to G$, its right trivialized tangent $\mbox{d}\tau_{\xi}\colon\mathfrak{g}\to\mathfrak{g}$ is defined as the $\xi$-dependent linear map obtained by composition of the linear maps
\[
\xymatrix@C=4pc{
\al\ar[r]^-{\{\xi\}\times \operatorname{id}}\ar@/_{4ex}/[rrr]_{\mathrm{d}\tau_\xi}
&\{\xi\}\times\al\ar[r]^{T_\xi\tau}&T_{\tau(\xi)}G\ar[r]^-{T_{\tau(\xi)}r_{\tau(\xi)^{-1}}}&T_eG\equiv\al
}
\]
where $r$ denotes right translation in the group. Since $\tau$ is a local diffeomorphism, all the arrows are linear isomorphisms. We denote the inverse of $\mbox{d}\tau_{\xi}$ as $\mbox{d}\tau_{\xi}^{-1}$. Omitting the first identification for brevity, we can write
\begin{align}
\mbox{d}\tau_{\xi}&=T_{\tau(\xi)}r_{\tau(\xi)^{-1}}\circ T_\xi\tau\label{eq:dtaudef}\\
\mbox{d}\tau_{\xi}^{-1}&=(T_\xi\tau)^{-1}\circ T_{e}r_{\tau(\xi)}=
T_{\tau(\xi)}(\tau^{-1})\circ T_{e}r_{\tau(\xi)}\label{eq:dtauinversedef}
\end{align}
\end{definition}
\begin{remark}
Omitting the identifications $\al\equiv\{\xi\}\times\al$, $\xi\in\al$, can lead to mismatches when using the definitions above explicitly; for example, if we rewrite equation \eqref{eq:dtauinverseAd} below using \eqref{eq:dtauinversedef}, then the left-hand side would be in $\{\xi\}\times\al$ while the right-hand side would be in $\{-\xi\}\times\al$. This should cause no problems if the identifications are made explicit when needed. In any case, \eqref{eq:dtauinverseAd} makes sense as an identity in $\al$. 
\end{remark}

\begin{lemma}\label{Ant}
(See \cite{MarsdenRatiu})
Let $g\in G$, $\lambda\in\al$ and $\delta f$ denote the variation of a function $f$ with respect to its parameters. Assuming $\lambda$ is constant, the following identity holds
\[
\delta(\Ad_{g}\,\lambda)=-\Ad_{g}\,[\lambda\,,\,g^{-1}\delta g],
\]
where $[\cdot\,,\,\cdot]:\al\times\al\Flder\al$ denotes the Lie bracket operation or equivalently $[\xi\,,\,\eta]\equiv\ad_{\xi}\eta$, for given $\eta,\,\xi\in\al$.
\end{lemma}

\begin{lemma} For each $\lambda\in\al$, the derivative of the map $\psi_\lambda\colon\al \to \al$ defined by $\psi_\lambda(\xi)=\Ad_{\tau(\xi)}\lambda$ is given by
\[
D\psi_\lambda(\xi)\cdot \eta=-[\Ad_{\tau(\xi)}\lambda\,,
\,\mathrm{d}\tau_{\xi}(\eta)],
\]
$\eta\in\al$. 
\end{lemma}
\begin{proof}
By lemma \ref{Ant},
\begin{align*}
D\psi_\lambda(\xi)\cdot \eta
&=-\Ad_{\tau(\xi)}[\lambda\,,\,\tau(\xi)^{-1}  T_\xi\tau(\eta)]\\
&=-[\Ad_{\tau(\xi)}\lambda\,,\, T_\xi\tau(\eta)\tau(\xi)^{-1}]\\
&=-[\Ad_{\tau(\xi)}\lambda\,,\,\mathrm{d}\tau_{\xi}(\eta)],
\end{align*}
obtained from the trivialized tangent definition \ref{Retr_again} and using the fact that $\Ad_{g}[\lambda\,,\,\eta]=[\Ad_{g}\lambda\,,\,\Ad_{g}\eta].$
\end{proof}
The lemma above holds not only for retraction maps but also for any smooth map $\tau\colon\al \to G$, for which $\mathrm{d}\tau_{\xi}$ can be defined as in definition \ref{Retr_again}.

The following lemma relates the right trivialized tangents at $\xi$ and $-\xi$, as well as their inverses.
\begin{lemma}\label{lemaPr_again}
For a retraction map $\tau\colon\al \to G$ and any $\xi,\eta\in\al$, the following identities hold:
\begin{align}
\label{eq:dtauAd}
\mathrm{d}\tau_{\xi}\,\eta&=\Ad_{\tau(\xi)}\,\mathrm{d}\tau_{-\xi}\,\eta,\\
\label{eq:dtauinverseAd}
\mathrm{d}\tau^{-1}_{\xi}\,\eta&=\mathrm{d}\tau_{-\xi}^{-1}\lp\Ad_{\tau(-\xi)}\,\eta\rp.
\end{align}
\end{lemma}
\begin{proof} Define $\rho(\xi)=\tau(\xi)^{-1}$. Differentiating and using definition \ref{Retr_again}, we get
\[
T\rho(\xi)\cdot \eta
=-T\ell_{\tau(\xi)^{-1}}\lp Tr_{\tau(\xi)^{-1}}\lp T\tau(\xi)\cdot\eta\rp\rp=
-T\ell_{\tau(\xi)^{-1}}\lp \mathrm{d}\tau_\xi(\eta)\rp,
\]
where $T\ell\,,\,Tr$ are the tangent of the left and right translations in the group respectively. On the other hand, we also have $\rho(\xi)=\tau(-\xi)$, so the chain rule implies
\[
T\rho(\xi) \cdot \eta=T\tau(-\xi)\cdot (-\eta)
= Tr_{\tau(-\xi)}(\mathrm{d}\tau_{-\xi}(-\eta))
= -Tr_{\tau(\xi)^{-1}}(\mathrm{d}\tau_{-\xi}(\eta))
\]
Equating both expressions we obtain \eqref{eq:dtauAd}.

For the second identity, replace $\eta$ by $\mathrm{d}\tau_{\xi}^{-1}\,\eta$ in \eqref{eq:dtauAd} to obtain
\[
\eta=\Ad_{\tau(\xi)}\,\mbox{d}\tau_{-\xi}\,\mbox{d}\tau_{\xi}^{-1}\,\eta.
\]
Solving for $\mbox{d}\tau_{\xi}^{-1}\,\eta$, we obtain \eqref{eq:dtauinverseAd}.
\end{proof}

\subsection*{Some retraction map choices}
\begin{itemize}

\item[a)] The exponential map $\e:\al\Flder G$, defined by $\exp(\xi)=\gamma(1)$, where $\gamma:\R\Flder G$ is the integral curve through the identity of the vector field associated with $\xi\in\al$ (hence, with $\dot\gamma(0)=\xi$). The right trivialized derivative and its inverse are
\begin{align*}
\mbox{d}\e_{x}\,y&=\sum_{j=0}^{\infty}\frac{1}{(j+1)!}\,\ad_{x}^{j}\, y,\\
\mbox{d}\e_{x}^{-1}\,y&=\sum_{j=0}^{\infty}\frac{B_{j}}{j!}\,\ad_{x}^{j}\, y,
\end{align*}
where $B_{j}$ are the Bernoulli numbers (see \cite{Hairer}). Typically, these expressions are truncated in order to achieve a desired order of accuracy.
\vspace{0.2cm}

\item[b)] The Cayley map $\ca:\al\Flder G$ is defined by $\ca(\xi)=(e-\frac{\xi}{2})^{-1}(e+\frac{\xi}{2})$ and is valid for a general class of {\it quadratic groups}. The quadratic Lie groups are those defined as
\[G=\lc Y\in GL(n, \R)\,\mid\,Y^{T}PY=P\rc,\]
where $P\in GL(n, \R)$ is a given matrix (here, $GL(n, \R)$ denotes the general linear
group of degree $n$). $O(n)$ or $SO(n)$ are examples of quadratic Lie groups. The corresponding Lie algebra is
\[\mathfrak{g}=\lc\Omega\in {\mathfrak gl}(n,
\R)\,\mid\,P\Omega+\Omega^T P=0\rc.\]
The right trivialized derivative and inverse of the Cayley map are defined by
\begin{align*}
\mbox{d}\ca_{x}\,y&=(e-\frac{x}{2})^{-1}\,y\,(e+\frac{x}{2})^{-1},\\
\mbox{d}\ca_{x}^{-1}\,y&=(e-\frac{x}{2})\,y\,(e+\frac{x}{2}).
\end{align*}
\end{itemize}

\subsection*{Applications to matrix groups: $SO(3)$}
We specify the exact form of the Cayley transform for the group $SO(3)$. While we have given more than one general choice for $\tau$, for computational efficiency we recommend the Cayley map since it is simple. In addition, it is suitable for iterative integration and optimization problems since its derivatives do not have any singularities that might otherwise cause difficulties for gradient-based methods. The group of rigid body rotations is represented by $3\times 3$ matrices with orthonormal column vectors corresponding to the axes of a right-handed frame attached to the body. Recall the map $\hat\cdot:\R^{3}\Flder\alg$ presented in \eqref{hatso}. A Lie algebra basis for $SO(3)$ can be constructed as $\lc\hat e_{1},\hat e_{2},\hat e_{3}\rc$, $\hat e_{i}\in\alg$, where $\lc e_{1},e_{2},e_{3}\rc$ is the standard basis for $\R^{3}$. Elements $\xi\in\alg$ can be identified with the vector $\omega\in\R^{3}$ through $\xi=\omega^{\alpha}\,\hat e_{\alpha}$, or $\xi=\hat\omega$. Under such identification the Lie bracket coincides with the standard cross product, i.e., $\ad_{\hat\omega}\,\hat\rho=\omega\times\rho$, for $\omega,\rho\in\R^{3}$. Using this identification we have
\[
\ca(\hat\omega)=I_{3}+\frac{4}{4+\parallel\omega\parallel^{2}}\lp\hat\omega+\frac{\hat\omega^{2}}{2}\rp,
\]
where $I_{3}$ is the $3\times 3$ identity matrix. The linear maps $\mbox{d}\tau_{\xi}$ and $\mbox{d}\tau_{\xi}^{-1}$ are expressed as the $3\times 3$ matrices
\[
\mbox{d}\ca_{\omega}=\frac{2}{4+\parallel\omega\parallel^{2}}(2I_{3}+\hat\omega),\,\,\,\,\,\mbox{d}\ca_{\omega}^{-1}=I_{3}-\frac{\hat\omega}{2}+\frac{\omega\,\omega^{T}}{4}.
\]

\end{document}